\documentclass[preprint]{elsarticle} 
\makeatletter
	\def\ps@pprintTitle{%
 	\let\@oddhead\@empty
	\let\@evenhead\@empty
	\def\@oddfoot{\centerline{\thepage}}%
	\let\@evenfoot\@oddfoot}
\makeatother

\usepackage[
            colorlinks=true,%
            breaklinks=true,%
            linkcolor=blue,
            urlcolor=blue,%
            citecolor=blue,%
            pdftitle={Quadratization of Autonomous Partial Differential Equations: Algorithmic Solutions}, 
            pdfkeywords={Keywords},	
            pdfauthor={Albani Olivieri, Gleb Pogudin, Boris Kramer},		
            bookmarksopen=false,
            pdfpagemode=UseNone]{hyperref}

\usepackage{makecell}

\usepackage[margin = 1.2in]{geometry}
\usepackage[T1]{fontenc}
\usepackage[english]{babel}
\usepackage[latin1]{inputenc} 
\usepackage{mathtools}
\usepackage{amsmath}
\usepackage{amsfonts}
\usepackage{amsthm}
\usepackage{amssymb}
\usepackage{array}
\usepackage{algorithm} 
\usepackage{algorithmicx}
\usepackage[noend]{algpseudocode}
\usepackage{multirow}
\usepackage{tikz}
\usepackage{subcaption}
\usepackage{siunitx}
\usepackage{forest}

\usepackage{color}
\usepackage{url}
\usepackage{cleveref}
\usepackage{graphicx}
\usepackage{breakcites}
\usepackage{soul} 
\usepackage{enumitem}
\usepackage[title,titletoc,toc]{appendix}

\makeatletter
\newenvironment{module}[1][htb]{%
    \renewcommand{\ALG@name}{Module}
   \begin{algorithm}[#1]%
  }{\end{algorithm}}
\makeatother

\usetikzlibrary{trees}

\usepackage[textsize=tiny]{todonotes}

\newtheorem{theorem}{Theorem}
\newtheorem{example}{Example}

\newtheorem{remark}{Remark}

\newtheorem{proposition}{Proposition}
\newtheorem{problem}{Problem}
\newtheorem{definition}{Definition}

\renewenvironment{proof}%
{\noindent {\textbf{Proof}:} }%
{\hfill $\Box$ \\[1ex] }

\newcommand{\bit}{\begin{itemize}}
\newcommand{\eit}{\end{itemize}}
\newcommand{\ben}{\begin{enumerate}}
\newcommand{\een}{\end{enumerate}}

\newcommand {\real} {\mathbb{R}}

\DeclareMathOperator*{\argmin}{arg\,min}%
\newcommand{\bI}{\ensuremath{\mathbf{I}}}

\newcommand{\bV}{\ensuremath{\mathbf{V}}}

\newcommand{\bb}{\ensuremath{\mathbf{b}}}

\newcommand{\bp}{\ensuremath{\mathbf{p}}}

\newcommand{\bu}{\ensuremath{\mathbf{u}}}
\newcommand{\bv}{\ensuremath{\mathbf{v}}}
\newcommand{\bw}{\ensuremath{\mathbf{w}}}
\newcommand{\bx}{\ensuremath{\mathbf{x}}}

\newcommand{\cC}{\ensuremath{\mathcal{C}}}
\newcommand{\cD}{\ensuremath{\mathcal{D}}}

\newcommand{\cG}{\ensuremath{\mathcal{G}}}

\newcommand{\cI}{\ensuremath{\mathcal{I}}}

\newcommand{\cM}{\ensuremath{\mathcal{M}}}

\newcommand{\cP}{\ensuremath{\mathcal{P}}}

\newcommand{\cR}{\ensuremath{\mathcal{R}}}
\newcommand{\cS}{\ensuremath{\mathcal{S}}}

\newcommand{\cV}{\ensuremath{\mathcal{V}}}
\newcommand{\cW}{\ensuremath{\mathcal{W}}}

\newcommand{\difforder}{differential order}

\graphicspath{{./figures/}}
\usetikzlibrary{trees,shapes,fit,calc}

\begin{document}

	\begin{frontmatter}
		
        \title{Quadratization of Autonomous Partial Differential Equations: \\Algorithmic Solutions}

		\author[ucsd]{Albani Olivieri\corref{cor1}}
            \ead{a1olivieri@ucsd.edu}
		\author[affil2]{Gleb Pogudin}
            \author[ucsd]{Boris Kramer}
			
		\cortext[cor1]{Corresponding author}
		
		\address[ucsd] {Department of Mechanical and Aerospace Engineering, University of California San Diego, CA, United States}
		
		\address[affil2] {LIX, CNRS, \'Ecole Polytechnique, Institute Polytechnique de Paris, France}
		
		\begin{abstract}
    Quadratization for partial differential equations (PDEs) is a process that formally transforms a PDE with a nonquadratic right-hand side into a quadratic form by introducing auxiliary variables. Even though the existence and uniqueness of the solution of this quadratic form are, as of yet, unknown in the general case, this symbolic transformation has been used in diverse fields to simplify the analysis, simulation, and control of PDE models. This paper presents a rigorous definition of PDE quadratization, a sample case study on the solutions of quadratic representations, and theoretical contributions for the PDE quadratization problem of spatially one-dimensional PDEs, including results on existence and complexity. Its main focus, however, is introducing and analyzing \texttt{QuPDE}, an algorithm based on symbolic computation and discrete optimization that outputs a quadratization for any spatially one-dimensional polynomial or rational PDE. This algorithm is the first computational tool to find quadratizations for PDEs to date. We demonstrate \texttt{QuPDE}'s performance by applying it to fourteen nonquadratic PDEs in diverse areas such as fluid mechanics, space physics, chemical engineering, and biological processes. \texttt{QuPDE} delivers a low-order quadratization in each case, uncovering quadratic transformations with fewer auxiliary variables than those previously discovered in the literature for some examples, and finding quadratizations for systems that had not been transformed to quadratic form before. 
		\end{abstract}	
		
		
\end{frontmatter}

\section{Introduction} \label{sec:intro}

The analysis and simulation of complex dynamical processes described by partial differential equations (PDE) remain a centerpiece of scientific inquiry. In particular, nonpolynomial and nonquadratic PDEs describe a wide range of dynamical processes in science and engineering. Some examples are the cubic FitzHugh-Nagumo model \cite{FHN}, which explains the activation and deactivation dynamics of a spiking neuron; the cubic Brusselator model \cite{brusselator}, used to predict oscillations in chemical reactions; the rational Euler equations that govern inviscid flows \cite{EULEREQ}; and the nonadiabatic tubular reactor model \cite{Heinemann1981} with an exponential Arrhenius reaction term. While moderate to high-degree polynomial and nonpolynomial models are common, their analysis, control, and model reduction are challenging. Multiple approaches for representing these systems have been explored to simplify their study, such as graphs \cite{Stewart2004, Nijholt2020}, linear formulations \cite{Brunton2016}, and quadratic recasting. Quadratic representations are particularly attractive among these methods, as nonlinear models with quadratic right-hand side are arguably easier to study, yet allow highly nonlinear dynamical behavior. 

For example, the works \cite{Kramer2021, Levin1994, Tesi, Zarei2018, Papachristodoulou2005} present techniques and guarantees for estimating stability domains for quadratic and polynomial systems, while \cite{Amato2010} addresses finite-time stability problems for such systems. In control theory, finding approximations of solutions and guarantees of feedback control problems for polynomial and quadratic systems has been the subject of multiple studies, e.g., \cite{Conger2022, Merola2023, Almubarak2019}. Recent techniques for solving Hamilton-Jacobi-Bellman equations, which provide an optimal feedback law for nonlinear systems, have shown that exploiting quadratic \cite{Borggaard2020, Kramer2024} and polynomial structure \cite{Corbin2024a, Kalise2018, Borggaard2021, Corbin2025} allows for scalable algorithm design to dimensions otherwise not possible. In \cite{Brenig2018}, the authors present theory and guarantees on stability, integrability, and solutions for quasi-polynomial dynamical systems, defined as linear combinations of monomials with real-valued exponents. The works \cite{Borri2023, Guillot2019, Guillot2019a} explore quadratic recasting to improve Taylor series methods for finding numerical solutions to dynamical and algebraic systems. Similarly, the study \cite{Bayer2024} uses this type of reformulation to perform stability analysis of nonlinear dynamical systems with periodic solutions. The authors in \cite{Savageau1987} also leverage variable transformations for converting nonlinear differential equations into S-systems, a canonical form that offers mathematical tractability and efficient numerical solutions. In \cite{Parker2000}, the authors exploit variable transformations to obtain polynomial forms to generate Taylor series solutions for PDEs through the Picard iteration.

Another scientific field that relies on quadratic representations is the area of reduced-order modeling. Transforming or lifting nonquadratic systems into quadratic ones was first introduced in \cite{GU2011} in the context of model reduction, and has since been used to obtain different---often quadratic---model reformulations for several purposes: to learn stable reduced models \cite{SKP_PIregulartizationOPINF, goyal2023guaranteed}, to perform projection-based model reduction \cite{KRAMER2019}, to develop system-theoretic reduced models \cite{BENNERBREITEN, KW2019_balanced_truncation_lifted_QB, Ritschel2020}, and for data-driven reduced-order modeling \cite{LIFT&LEARN, McQuarrie2021, Guo2022, McQuarrie2023, Khodabakhshi2022, Swischuk2020, Qian2019, Jain2021}. Analog computation with chemical reaction networks has also exploited quadratic model reformulation. By lifting differential equation models into quadratic form, a chemical reaction network can be transformed into a bimolecular one to obtain an elementary chemical reaction representation. This result was fundamental in proving the Turing completeness of elementary chemical reactions \cite{hemery2020complexity, TuringCompleteness, BOURNEZ2007317}. 

The works mentioned above can also be characterized by the type of system they focus on. In particular, \cite{Levin1994, Tesi, Zarei2018, Papachristodoulou2005, Amato2010, Merola2023, Almubarak2019, Brenig2018, Borri2023, Guillot2019, Guillot2019a, Savageau1987, GU2011, Ritschel2020, hemery2020complexity, TuringCompleteness, BOURNEZ2007317, Bayer2024} address ODE dynamical systems, \cite{Conger2022} focuses on discrete-time dynamical systems, \cite{Kramer2021, Kramer2024, Borggaard2020, Corbin2024a, Borggaard2021, Corbin2025, SKP_PIregulartizationOPINF, goyal2023guaranteed, BENNERBREITEN, KW2019_balanced_truncation_lifted_QB, McQuarrie2021} apply their methods on semi-dicretized PDEs, and \cite{Kalise2018, Parker2000, KRAMER2019, LIFT&LEARN, Guo2022, McQuarrie2023, Khodabakhshi2022, Swischuk2020, Qian2019, Jain2021} leverage quadratic and polynomial representations for PDEs in their continuous form.

Throughout this paper, we call a \textit{quadratization} the set of auxiliary functions that one may need to add to obtain a quadratic form of a PDE system with higher-degree polynomial (or nonpolynomial) right-hand side expressions. The quadratic reformulation of the PDE will comprise the rewritten form of the original differential equations plus the corresponding differential equations for the auxiliary functions. Moreover, we treat the quadratization problem as a symbolic one. Therefore, we refer to the PDEs right-hand side functions as variables, including those introduced by a quadratization. We call a quadratization optimal---see Section~\ref{sec:theory} for a formal definition---if it introduces the minimum number of auxiliary variables, and a monomial quadratization if the variables introduced are one-term polynomials. Let us illustrate PDE quadratization with an example. 

\begin{example}
    Consider the PDE for the unknown function $u(x,t)$ as
\begin{equation}
\label{eq:intro-ex}
    \dfrac{\partial u}{\partial t} = \dfrac{\partial u}{\partial x} + u^3.
\end{equation}
Let the initial boundary value problem be defined by the initial condition $u(x,0)=u_0(x)$ and the boundary conditions $u(a,t)=c$ and $u(b,t)=0$ on the domain $[a,b]$, with $a,b,c\in\real$ and $b> a$. 
To bring~\eqref{eq:intro-ex} into quadratic form, we formally define the auxiliary variable $w = w(u)\coloneq u^2$. We now need to add a new differential equation to~\eqref{eq:intro-ex}: 
\begin{equation}
\label{eq:intro-ex-2}
    \frac{\partial w}{\partial t} = 2u\frac{\partial u}{\partial t}= 2u\frac{\partial u}{\partial x} + 2u^4.
\end{equation} 

Using the definition of $w$, we write the new quadratic PDE system as
\begin{align}
\label{eq:quad-intro}
    \dfrac{\partial u}{\partial t} =\dfrac{\partial u}{\partial x} + uw \quad \textnormal{ and } \quad \dfrac{\partial w}{\partial t} = 2u\frac{\partial u}{\partial x} + 2w^2,
\end{align}
with the initial conditions for the new quadratic form given by $u(x,0)= u_0(x),$ $w(x,0)=u_0(x)^2$, and the boundary conditions $u(a,t)=c$, $u(b,t)=0$, $w(a,t)=c^2$ and $w(b,t)=0$. Since the right-hand side expressions in \eqref{eq:quad-intro} are quadratic in the states $\{u, w\}$, we say that the set $\{u^2\}$ is a quadratization for the PDE~\eqref{eq:intro-ex}.
\end{example}

There is a lack of computational tools for finding quadratizations for PDE systems; so far, they have been derived by hand. There are several reasons for that. First, finding an optimal quadratization is an NP-hard problem---see \cite{hemery2020complexity} for the case of ODEs and our result in Section~\ref{sec:exisNP} for the PDE case.
Second, due to the multivariate nature of PDEs, the spatial derivatives of the unknown functions and auxiliary variables can become part of the quadratic transformations. As a result, the process becomes more complex, since the quadratization should include the derivatives of each auxiliary variable up to a defined order depending on the smoothness of the solution.
These challenges have made it difficult to develop computational methods and theoretical guarantees for quadratization of PDEs---in contrast to the ODE setting, where multiple studies have focused on the quadratization problem, and two software tools are available to compute a quadratization for any ODE \cite{quadODE, hemery2020complexity, Carothers2005, GU2011, bychkov2023exact}. 

In \cite{quadODE}, the authors present the \texttt{QBee} algorithm that, given a nonquadratic polynomial ODE system, finds an optimal quadratic transformation by introducing auxiliary variables that are monomials in the original variables. An extension to \texttt{QBee} was introduced in the later work \cite{bychkov2023exact} to handle a special case of quadratization for arbitrary-dimensional ODEs. This type of transformation, referred to as dimension-agnostic quadratization, brings to quadratic form a class of ODEs that maintain the same symbolic structure even as their dimension varies and yields a quadratization that can be generalized to any dimension. An example of such an ODE system is a semi-discretized PDE that is affine in its spatial derivatives, where the nonlinear terms remain symbolically identical when the discretization size increases, and the coupling term from the approximation of the spatial derivatives is linear. With this extension of \texttt{QBee}, one can find quadratizations for such PDEs by computing their semi-discretization as a preprocessing step. The work in \cite{bychkov2023exact} is the only study that could partially address the problem of finding quadratizations for a specific class of PDEs. The shortcomings of this approach are: (i) The spatial discretization of a PDE can be time-consuming and problem-specific; (ii) \texttt{QBee} only produces quadratizations for a subclass of PDEs, the ones that are affine in the spatial derivatives; and (iii) \texttt{QBee} guarantees optimal monomial quadratizations only for ODE systems. For dimension-agnostic quadratizations, it does not guarantee optimality (\cite[Section~5.5]{bychkov2023exact}), which can lead to large quadratic PDE systems after lifting (see Section~\ref{sec:res-qbee}). 

Thus, to the best of our knowledge, no algorithm for finding quadratizations for polynomial or rational PDEs exists to date, despite the high interest in many fields, as referenced above, in using quadratically recasted models. Moreover, no theory has been developed on the PDE quadratization problem or the resulting quadratic models. To address these gaps, the main contributions of this work are (i) initial theoretical groundwork for the quadratization of spatially one-dimensional PDEs by providing a formal definition, a result on the solutions of the quadratic representation of a particular PDE, and theory about the underlying algorithmic problem, and (ii) an algorithm based on symbolic computation and discrete optimization techniques that outputs a low-order quadratization for any spatially one-dimensional PDE with polynomial or rational nonlinearities on its right-hand side. 

This paper is structured as follows: Section~\ref{sec:theory} shows our theoretical contributions related to the quadratization problem applied to PDEs. Section~\ref{sec:alg} presents the algorithm \texttt{QuPDE}, which produces quadratizations for spatially one-dimensional polynomial and rational PDEs. Section~\ref{sec:bench-res} demonstrates the broad applicability of the \texttt{QuPDE} algorithm on fourteen nonlinear and nonquadratic PDEs from multiple fields, including chemistry, space physics, biology, fluid dynamics, and engineering. We conclude with a summary and an outlook on topics of further research in Section~\ref{sec:conclusions}.

\section{Theoretical Results} \label{sec:theory}
In Section~\ref{sec:notation}, we introduce the notation used throughout the paper and frame the scope of this work. Then, we present a rigorous definition for the quadratization of a spatially one-dimensional PDE in Section~\ref{sec:def}. Section~\ref{sec:ex-mkdv} shows how solutions of a lifted system obtained by a quadratization are related to the original PDE for the modified Korteweg-de Vries equation. In Section~\ref{sec:exisNP}, we show theoretical results related to the existence of such transformations and the complexity of finding an optimal one.

\subsection{Notation and limitations}
\label{sec:notation}
Let $\Omega\subseteq\mathbb{R}$ denote the spatial domain and $x\in \Omega$ a point in it. We use $C^d(\Omega)$ to refer to the space of functions that are $d$-times continuously differentiable on the set $\Omega$. 
When referring to PDE systems, we denote by $\bu:\Omega\times[0,\infty) \rightarrow\mathbb{R}^n$ a vector of functions such that $u_i: \Omega\times[0,\infty)\rightarrow \mathbb{R}$, where $i=1,\dots, n$. For brevity, we often omit the explicit dependence of $\bu$ on the spatial and temporal variables, and write $\bu = \bu(x,t)$. Moreover, we use $\frac{\partial u}{\partial x} = \partial_xu = u_x$ interchangeably to refer to the partial derivatives of $u$; for higher-order derivatives, we write $\frac{\partial^h u}{\partial x^h} = \partial^h_xu=u_{x\dots x}.$ We denote the vector of the $x$-partial derivatives of $\bu$ of order $h$ as $\partial_x^h\bu$, i.e.,
   \[
    \partial_x^h\bu := \begin{bmatrix}
        \partial^h_xu_1 & \dots &  \partial^h_xu_n
    \end{bmatrix}^\top.
    \]
To formulate the quadratization problem as a symbolic computing problem, one with explicit mathematical expressions and with exact input and outputs, let $\mathbb{F}$ denote a field, and consider $\upsilon_1, \ldots, \upsilon_n$ \textit{formally} as symbolic variables that could refer to functions or their partial derivatives, e.g., $\upsilon_1:=u(x, t)$, $\upsilon_2:=u_x(x,t)$ or $\upsilon_3:=u_{xx}(x,t)$. We use the notation $\mathbb{F}[\upsilon_1,\dots, \upsilon_n]$ to refer to the set of all polynomials in the variables $\upsilon_j$ with $j=1,\dots, n$. More generally, let $\boldsymbol{\upsilon}_1, \dots, \boldsymbol{\upsilon}_m$ be vectors with $\boldsymbol{\upsilon}_i:=\begin{bmatrix}
        \upsilon_1^{(i)} & \dots &  \upsilon_n^{(i)}
    \end{bmatrix}^\top$, for all $i=1, \dots, m$. Then, we refer to the set of polynomials in the variables $\upsilon_j^{(i)}$ as
    \[\mathbb{F}[\boldsymbol{\upsilon}_1, \dots, \boldsymbol{\upsilon}_m]=\mathbb{F}[\upsilon^{(1)}_1,\dots, \upsilon^{(1)}_n, \dots, \upsilon^{(m)}_1,\dots, \upsilon^{(m)}_n].\]
    Let $\bv$ be a vector with $d$ polynomial entries in $\mathbb{F}[\upsilon_1,\dots, \upsilon_n]$. Then we write $\bv\in\mathbb{F}^d[\upsilon_1,\dots, \upsilon_n]$, where $\mathbb{F}^d[\upsilon_1,\dots, \upsilon_n]$ represents the space of all vector-valued functions with $d$ components in $\mathbb{F}[\upsilon_1,\dots, \upsilon_n]$. We call a \textit{monomial} a polynomial that has only one term, and the \textit{total degree} of a monomial is defined as the sum of the degrees of all its variables.

Let $w_1\in \mathbb{R}[\bu, \partial_x\bu, \dots, \partial^{c_1}_x\bu], \dots, w_\ell\in \mathbb{R}[\bu, \partial_x\bu, \dots, \partial^{c_\ell}_x\bu]$ and $\bw := \begin{bmatrix}
    w_1 & \dots & w_{\ell}
\end{bmatrix}$. For each $i=1,\dots, \ell$, let $c_i$ denote the maximum spatial-derivative order appearing in $w_i$, and set $\cC:=\{c_1, \dots, c_\ell\}$; for example, if $w_1, \dots, w_{\ell}$ have the same maximum spatial-derivative order equal to 1, then $c_1 = \dots = c_\ell = 1$, and $\cC:=\{1, \dots, 1\}$. For each polynomial $w_i$ and a fixed integer $k$, we define $\partial_x^{\leq k-c_i}\bw_i$ as the vector with the spatial derivatives of $w_i$ up to order $k-c_i$, i.e.,
\[\partial_x^{\leq k-c_i}\bw_i: = \begin{bmatrix}
    w_i & \partial_xw_i & \dots & \partial_x^{k-c_i}w_i
\end{bmatrix}^\top.\]
Using this notation, we define the vector
\begin{equation}
\label{eq:notation-ord}
    \partial_x^{\leq k-\cC}\bw: = \begin{bmatrix}
    \partial_x^{\leq k-c_1}\bw_1 & \dots & \partial_x^{\leq k-c_{\ell}}\bw_{\ell}
\end{bmatrix}^\top,
\end{equation}
 which contains the spatial derivatives of the variables in $\bw$ up to the order determined by the set $\cC$ and the integer $k$. 

\begin{remark}
\label{rem:scope}
    Before showing our theoretical results on PDE quadratization, we emphasize that this work does not focus on the analysis of certain aspects of PDE theory, such as general results on the existence and uniqueness of solutions to the boundary value problem associated with the quadratic PDE reformulations and its properties, such as stability. While this analysis is important and remains a direction for future research, the scope of this work is to construct and analyze an algorithm to obtain new symbolic (quadratic) representations for PDEs that are not intended to be solved or discretized. Instead, these transformations, which we derive from the PDE's symbolic structure, can be exploited for model learning \cite{LIFT&LEARN, Netto2021, Lew2023}, system analysis \cite{Kramer2021}, and control \cite{Kramer2024}.

\end{remark}

\subsection{Definitions} \label{sec:def}

We offer a rigorous definition for quadratization of a spatially one-dimensional evolutionary PDE.

\begin{definition}(Quadratization for PDEs)
    \label{def:pde-quad}
        Let $\bu \coloneq \begin{bmatrix}
            u_1 & \dots & u_n
        \end{bmatrix}^\top$ be a vector of unknown functions $u_1(x,t), \dots, u_n(x,t)$. Consider a polynomial, spatially one-dimensional,         evolutionary PDE system of the form
        \begin{equation}
        \label{eq:def-pde}
            \partial_t u_{1} = p_1(\bu, \partial_x\bu, \dots, \partial^h_x\bu), \hspace{5px} \dots, \hspace{5px} \partial_t u_{n} = p_n(\bu, \partial_x\bu, \dots, \partial^h_x\bu),
        \end{equation}
with $p_1, \dots, p_n \in \mathbb{R}[\bu, \partial_x\bu, \dots, \partial^h_x\bu]$. 
For polynomials $$w_1\in \mathbb{R}[\bu, \partial_x\bu, \dots, \partial^{c_1}_x\bu], \dots, w_\ell\in \mathbb{R}[\bu, \partial_x\bu, \dots, \partial^{c_\ell}_x\bu],$$ define the vector 
$\bw := \begin{bmatrix}
    w_1 & \dots & w_{\ell}
\end{bmatrix}^\top$
and the set 
\begin{equation*}
\label{eq:quad-vars}
    \cW := \{w_1, \dots, w_{\ell}\}.
\end{equation*}

Then, $\cW$ is said to be a \textnormal{quadratization} of \textnormal{\difforder} $k$ for the PDE~\eqref{eq:def-pde}, 
if there are polynomials $g_1, \dots, g_{n+\ell}\in \mathbb{R}[\bu, \partial_x\bu, \dots, \partial^{k}_x\bu, \partial_x^{\leq k-\cC}\bw]$ of total degree at most two, such that 
\begin{equation}
    \label{eq:def-quad}
\begin{aligned}
    \partial_tu_{1} = g_1, \hspace{5px} &\dots, \hspace{5px} \partial_tu_{n} = g_n, \\
    \partial_tw_{1} = g_{n+1}, \hspace{5px} &\dots, \hspace{5px}\partial_tw_{\ell}= g_{n+\ell},
\end{aligned}
\end{equation}
for every $u_1, \dots, u_n\in C^{k}(\Omega)$. Here, the notation $\partial_x^{\leq k-\cC}\bw$ refers to the vector of partial derivatives introduced in~\eqref{eq:notation-ord} with $\cC=\{c_1, \dots, c_\ell\}$. The cardinality of $\cW$ is called the \textnormal{order of quadratization}, while the differential order $k$ refers to the maximal spatial-derivative order of $\bu$ in the resulting quadratic PDE system. A quadratization of the smallest possible order is called an \textnormal{optimal quadratization}. If all polynomials $w_1,\dots, w_\ell$ are monomials, the quadratization is called a \textnormal{monomial quadratization}. 
\end{definition}

\subsection{Solutions of quadratic representations} \label{sec:ex-mkdv}
To better understand the solutions of the PDE lifted representations obtained through quadratization variables, and their relation with the original PDE solutions, we present Proposition~\ref{prop:sol-lifted-mkdv}. This result, illustrated on the modified Korteweg-de Vries (KdV) equation, shows the relation between the solutions of the quadratic representation and those of the original PDE, and establishes conditions to conclude existence and uniqueness for the initial value problem of the lifted system.

The modified KdV equation is an integrable nonlinear PDE that describes acoustic waves in anharmonic lattices and Alfv\'en waves in collisionless plasma~\cite{Wadati1973}. In this analysis, we consider this PDE in the general form
\begin{equation}
\label{eq:kdv-sol-proof}
    u_{t}=-6u^{2}u_x -u_{xxx}
\end{equation}
with the initial condition $u_0(x):=u(x,0)$. This equation admits a known family of real-valued solutions $u(x,t)\in C^{\infty}(\mathbb{R})$~\cite{Wadati1973}, which can be obtained through the inverse scattering method. Moreover, it has unique strong solutions in a defined interval if $u_0\in H^s(\mathbb{R})$ with $s\geq 1/4$~\cite{Linares2015}.  

Consider the quadratization $\{u^2\}$ for~\eqref{eq:kdv-sol-proof}, which we obtained by applying our algorithm \texttt{QuPDE} (see Section~\ref{sec:res} for computational details). We now derive a quadratic representation of~\eqref{eq:kdv-sol-proof} defining the new variable $w(x,t):=u^2$:
\begin{align}
    \label{eq:lifted-mkdv-1}
    v_t &= -3w_xv-v_{xxx}, \\
    \label{eq:lifted-mkdv-2}
    w_t &= -2v_{xxx}v - 6w_xw.
\end{align}
Here, we rename the state $u(x,t)$ as $v(x,t)$ to emphasize that the lifted system is studied independently of the original equation. The following proposition relates the solutions of the system~\eqref{eq:lifted-mkdv-1}-\eqref{eq:lifted-mkdv-2} to the original PDE.
\begin{proposition}
\label{prop:sol-lifted-mkdv}
    Let $\ell\geq3$ and $u(x,t)\in C^\ell(\mathbb{R})$ be a solution for~\eqref{eq:kdv-sol-proof} with the initial condition $u_0(x):=u(x,0)$. If $v_0(x):=v(x,0)=u_0(x)$ and $w_0(x):=w(x,0)=u_0^2(x)$, then the pair 
\begin{equation}
\label{eq:uniq-sol}
    (u(x,t), u^2(x,t))
\end{equation} 
solves the lifted system~\eqref{eq:lifted-mkdv-1}-\eqref{eq:lifted-mkdv-2} for $t>0$ and $x\in\mathbb{R}$. Moreover, if $u_0(x)\in H^s(\mathbb{R})$ with $s\geq 1/4$, then the solution $u(x,t)$ is unique, and hence the lifted system has a unique solution given by~\eqref{eq:uniq-sol}.
\end{proposition}
\begin{proof}
    First, note that the pair~\eqref{eq:uniq-sol} is, by construction of the quadratic lifting and with the choice of $w_0(x)=u_0^2(x)$, a solution to~\eqref{eq:lifted-mkdv-1}-\eqref{eq:lifted-mkdv-2}. Now, we will show that the conservation law
    \begin{equation}
    \frac{\partial}{\partial t}(w-v^2) = 0
\end{equation}
holds for the entire interval of existence. We introduce $E=E(x,t):=w-v^2$ and differentiate it with respect to time: 
\begin{align}
    \frac{\partial E}{\partial t} &= \frac{\partial}{\partial t}(w-v^2) =  w_t -2v_tv \\
    & = -2v_{xxx}v - 6w_xw +6w_xv^2 +2v_{xxx}v \\
    & = 6w_x(v^2-w) \\
    & = -6w_xE.
\end{align}
For a fixed $x$, the equation above is a scalar linear ODE in $t$ with the solution: 
\begin{equation}
    E(x,t) = E(x,0)e^{-6\int_0^tw_x(x,\tau)d\tau}.
\end{equation}
Note $E(x,0) = w_0(x) -v_0^2(x)= u_0^2(x) - u_0^2(x) =0$. Therefore, the conservation law $\frac{\partial}{\partial t}(w-v^2) = 0$ holds for all time, and $w=v^2$ as long as a solution exists.

We substitute $w=v^2$ into~\eqref{eq:lifted-mkdv-1}: 
\begin{align}
    v_t &= -3(v^2)_xv-v_{xxx} = -6v^2v_x-v_{xxx},
\end{align}
which yields the modified Korteweg-de Vries equation. Therefore, $v(x,t)$ is a solution to the original equation's initial value problem:
\begin{equation}
\label{eq:IVP}
    \begin{cases}
    v_t +6v^2v_x+v_{xxx}=0, \\
    v(x,0)=v_0(x).
\end{cases}
\end{equation}
\noindent Moreover, if~\eqref{eq:IVP} satisfies the conditions in~\cite[Theorem~7.1]{Linares2015}, then $v(x,t)$ is a unique strong solution to the IVP~\eqref{eq:IVP}. Additionally, with $v_0(x)=u_0(x)$, we obtain $u=v$ and $w=u^2$ throughout the interval defined in \cite[Theorem 7.1]{Linares2015}. By the uniqueness of $u(x,t)$, we conclude that the solution $w(x,t)$ is also unique, yielding a unique solution for the lifted system~\eqref{eq:lifted-mkdv-1}-\eqref{eq:lifted-mkdv-2}.
\end{proof}

The first statement of Proposition~\ref{prop:sol-lifted-mkdv} could be generalized to establish the existence of solutions for quadratic representations, given sufficient regularity of the original solutions. 

\begin{remark}
    By construction, if the PDE system~\eqref{eq:def-pde} has existing solutions $u_1, \dots, u_n\in C^{k}(\Omega)$, then its quadratic representation~\eqref{eq:def-quad} should have a solution comprised of $\bu$ and the lifting map $\bw$, with initial conditions given by the image of the original initial data under the lifting map.
\end{remark}

\subsection{Existence and NP-hardness} \label{sec:exisNP}

We provide a theoretical result regarding the existence of quadratizations for polynomial PDEs. Then, we discuss the complexity of finding an optimal quadratization.  

    \begin{theorem}[Existence of a PDE quadratization]
    \label{thm:existence}
    A PDE system of the form~\eqref{eq:def-pde} of order $h$ has a monomial quadratization of \textit{\difforder} $3h$, given $u_1,\dots, u_n\in C^{\ell}(\Omega)$, where $\ell \geq3h$. 
    \end{theorem}
    \begin{proof}
Since we look for a quadratization with \textnormal{\difforder} $3h$, we assume that $u_1, \dots, u_n$ are sufficiently smooth functions to compute $2h$ spatial derivatives of the original system~\eqref{eq:def-pde}. Computing these derivatives results in the following equations
\begin{align*}
\partial_t u_{1} &= p_1(\bu, \partial_x\bu, \dots, \partial^h_x\bu), & \quad &\dots, & \quad \partial_t u_{n} &= p_n(\bu, \partial_x\bu, \dots, \partial^h_x\bu), \\
\partial_x(\partial_t u_{1}) &= \partial_xp_1(\bu, \partial_x\bu, \dots, \partial^h_x\bu), & \quad & \dots, & \quad \partial_x(\partial_t u_{n}) &= \partial_xp_n(\bu, \partial_x\bu, \dots, \partial^h_x\bu), \\
& \quad & \vdots & \quad &\\
\partial_x^{2h}(\partial_t u_{1}) &= \partial_x^{2h}p_1(\bu, \partial_x\bu, \dots, \partial^h_x\bu), & \quad & \dots, &\quad \partial_x^{2h}(\partial_t u_{n}) &= \partial_x^{2h}p_n(\bu, \partial_x\bu, \dots, \partial^h_x\bu).
\end{align*}
Since $u_1,\dots, u_n\in C^{3h}(\Omega)$, we can apply Schwarz's theorem to interchange the order of differentiation. Then, we write the equivalent system 
\begin{equation}
\label{eq:proof-quad}
\begin{aligned}
    \partial_t u_{1} & = p_1(\bu, \partial_x\bu, \dots, \partial^h_x\bu), & \quad &\dots, & \quad \partial_t u_{n} & = p_n(\bu, \partial_x\bu, \dots, \partial^h_x\bu), \\
\partial_t(\partial_x u_{1}) &= \partial_xp_1(\bu, \partial_x\bu, \dots, \partial^h_x\bu), & \quad & \dots, & \quad  \partial_t(\partial_x u_{n}) & = \partial_xp_n(\bu, \partial_x\bu, \dots, \partial^h_x\bu), \\
& \quad & \vdots & \quad &\\
\partial_t(\partial_x^{2h} u_{1}) &= \partial_x^{2h}p_1(\bu, \partial_x\bu, \dots, \partial^h_x\bu), & \quad & \dots, &\quad \partial_t(\partial_x^{2h} u_{n}) &= \partial_x^{2h}p_n(\bu, \partial_x\bu, \dots, \partial^h_x\bu).
\end{aligned}
\end{equation}

By treating each $\partial_x^i u_j$ as a function of $t$ (e.g., by fixing an arbitrary value of $x$), we can symbolically consider \eqref{eq:proof-quad} as an ODE system with $n(2h + 1)$ state variables $\mathbf{u}, \partial_x \mathbf{u}, \ldots, \partial_x^{2h} \mathbf{u}$ and $nh$ input variables $\partial_x^{2h+1} \mathbf{u}, \ldots, \partial_x^{3h} \mathbf{u}$. Now, let us write the definition of an input-free quadratization from \cite[Definition~3.4]{bychkov2023exact}.

\begin{definition}
    Consider the system
    \begin{equation}
    \label{eq:def3.4}
        \frac{\textnormal{d}\bx}{\textnormal{d}t} = \bp(\bx, \bu)
    \end{equation}
    where $\bx =\bx(t) = \left[x_1, \ldots, x_N\right]^{\top}$ are the states, $\bu =\bu(t) = \left[u_1, \ldots, u_r\right]^{\top}$ denote the inputs and $\bp\in\mathbb{C}^N[\bx,\bu]$. An $\ell$-dimensional vector of new variables 
    $$\bw = \bw(\bx)\in\mathbb{C}^\ell [\bx]$$ is said to be an \textnormal{input-free quadratization} of~\eqref{eq:def3.4} if there exist vectors $\mathbf{q}_1(\mathbf{x}, \mathbf{w}, \mathbf{u})$ and $\mathbf{q}_2(\mathbf{x}, \mathbf{w}, \mathbf{u})$ of dimensions $N$ and $\ell$, respectively, with the entries being polynomials of total degree at most two such that
    $$
\frac{\textnormal{d}\bx}{\textnormal{d}t}=\mathbf{q}_1(\mathbf{x}, \mathbf{w}, \mathbf{u}) \quad \textit{ and } \quad \frac{\textnormal{d}\bw}{\textnormal{d}t}=\mathbf{q}_2(\mathbf{x}, \mathbf{w}, \mathbf{u}).
$$
\end{definition}

Our goal is to prove that~\eqref{eq:proof-quad} has an input-free quadratization. We begin by stating~\cite[Proposition~3.9]{bychkov2023exact}.
\begin{proposition}
\label{prop:prop3.9}
    Let $\bx = \bx(t)=\begin{bmatrix}x_1(t)& \dots &x_N(t)\end{bmatrix}^\top$. Consider an input-affine system of the form 
    \begin{equation}
    \label{eq:prop3.9}
        \frac{\textup{d}{\mathbf{x}}}{\textup{d}t}=\bp_0(\mathbf{x})+\sum_{i=1}^r \bp_i(\mathbf{x}) g_i,
    \end{equation}
where $g_1, \dots, g_r$ are input functions and $\bp_0(\bx)\dots\bp_r(\bx)$ are vectors of polynomials. We introduce $r$ differential operators:
$$
D_i(x_j):=\mathbf{p}_i(\mathbf{x})^{\top} \cdot \frac{\partial x_j}{\partial\mathbf{x}}, \quad 1 \leq i \leq r, \quad 1\leq j \leq N
$$
where $\frac{\partial}{\partial \mathbf{x}}=\left[\frac{\partial}{\partial x_1}, \ldots, \frac{\partial}{\partial x_N}\right]^{\top}$. Let $\mathcal{A}$ be a subalgebra generated by $D_1, \ldots, D_r$ in the algebra $\mathbb{C}\left[\mathbf{x}, \frac{\partial}{\partial \mathbf{x}}\right]$ of all polynomial differential operators in $\mathbf{x}$. Then there is an input-free quadratization for~\eqref{eq:prop3.9} if and only if
\begin{equation}
\label{eq:cond-proof}
    \operatorname{dim}\left\{A\left(x_j\right) \mid A \in \mathcal{A}\right\}<\infty \quad \text { for every } 1 \leq j \leq N.
\end{equation}
\end{proposition}
Therefore, to prove the existence of an input-free quadratization, we show that the system~\eqref{eq:proof-quad} satisfies the assumptions of Proposition~\ref{prop:prop3.9}. First, we prove that~\eqref{eq:proof-quad} is input-affine by contradiction. Let us assume that~\eqref{eq:proof-quad} is not input-affine. This means that at least one of the right-hand side expressions in~\eqref{eq:proof-quad} has a monomial with $m$ spatial derivatives of order greater than $2h$, where $m\geq2$. Initially, these derivatives had an order of at most $h$, so each must have been differentiated at least $2h+1-h = h+1$ times. Each differentiation increases only one of the derivative orders in a monomial by one; therefore, the right-hand side that is not input-affine was differentiated at least $m(h+1)=mh+m$ times. However, to obtain the system~\eqref{eq:proof-quad} only $2 h$ differentiations were taken, and since $2h< mh + m$, we arrive at a contradiction. Then, we conclude that~\eqref{eq:proof-quad} must be input-affine. 

To prove the condition in~\eqref{eq:cond-proof} for~\eqref{eq:proof-quad}, we rely on \cite[Lemma~3.13]{bychkov2023exact}. This lemma states that an ODE with state variables $\upsilon_1,\dots, \upsilon_N$ meets the condition~\eqref{eq:cond-proof} if its corresponding differential operators satisfy 
\begin{equation}
\label{eq:diff-op-cond}
    D_i(\upsilon_j)\in \mathbb{C}[\upsilon_1, \dots, \upsilon_{j-1}]
\end{equation} 
for every $1\leq i\leq r$ and $1\leq j\leq N$.  For system~\eqref{eq:proof-quad}, we take $N =n(2h + 1)$ and identify $\upsilon_1, \dots, \upsilon_{N}$ with the symbolic state variables $\mathbf{u}, \partial_x \mathbf{u}, \ldots, \partial_x^{2h} \mathbf{u}$ ordered, so that lower-order derivatives precede higher-order ones.

Now, consider an equation in~\eqref{eq:proof-quad} with a left-hand side of the form $\partial_t(\partial_x^au_j)$ and the input variable $\partial_x^bu_k$ in its right-hand side, with $1\leq j,k\leq n$, $a\leq2h$ and $2h+1\leq b \leq 3h$. In this equation, the monomial with the partial derivative $\partial_x^bu_k$ was differentiated at least $b-h$ times. Since the system is input-affine, it follows that any other derivative in this monomial has order of at most $h+(a-(b-h))$. To have differential operators that satisfy~\eqref{eq:diff-op-cond}, the inequality $h+(a-(b-h))<a$ must hold for every $a$ and $b$, which is equivalent to $2h<b$. Since $2h+1\leq b$ is always true in system~\eqref{eq:proof-quad}, the condition in~\eqref{eq:cond-proof} is met. With this, we deduce that the ODE system~\eqref{eq:proof-quad} has a monomial input-free quadratization. Note that the input-free monomial quadratization of~\eqref{eq:proof-quad} is also a quadratization of the original PDE~\eqref{eq:def-pde} with a \textnormal{\difforder} of $3h$.
\end{proof}
In \cite{hemery2020complexity}, the authors showed that finding optimal polynomial quadratizations of ODEs is NP-hard. Since PDEs are more general mathematical structures than ODEs, we present a reduction-based proof to show that the same problem for PDEs is of the same complexity.
\begin{proposition}
    Finding an optimal monomial quadratization for a nonquadratic polynomial system of PDEs of the form \eqref{eq:def-pde} is NP-hard.  
\end{proposition}
\begin{proof}
    To prove by reduction, let us reformulate the ODE quadratization problem such that it becomes an instance of the PDE quadratization problem. Consider a polynomial system of ODEs
    \begin{equation}
    \label{eq:thm-ode}
        \dot{u}_1 = p_1(u_1, \dots, u_n), \quad \dots, \quad \dot{u}_n = p_n(u_1, \dots, u_n).
    \end{equation}
    Defining $\bu(t) := \begin{bmatrix}
        u_1(t) & \dots & u_n(t)
    \end{bmatrix}^\top$, we state the equivalence $\dot{\bu}=\partial_t{\bu}$, which shows that the ODE system~\eqref{eq:thm-ode} is equivalent to~\eqref{eq:def-pde} with $h=0$, i.e., with no partial derivatives in the space variable. Then, finding an optimal monomial quadratization for~\eqref{eq:thm-ode} can be reduced to solving the same problem for the PDE 
    \begin{equation*}
        \partial_t{u_1} = p_1(u_1, \dots, u_n), \quad \dots, \quad \partial_t{u_n} = p_n(u_1, \dots, u_n).
    \end{equation*}
    Since finding an optimal monomial quadratization for an ODE is an NP-hard problem \cite{hemery2020complexity}, we conclude that the same problem applied to PDEs is also NP-hard.
\end{proof} 

\section{PDE Quadratization Algorithm}
\label{sec:alg}
Having presented a theoretical foundation of PDE quadratizations, we introduce the first algorithmic solution to find such transformations. We provide a detailed description of the proposed algorithm in Section~\ref{sec:outline}, and discuss its practical aspects in Section~\ref{sec:correct}.

\subsection{Proposed algorithm: \texttt{QuPDE}}
\label{sec:outline}

We present \texttt{QuPDE}, a computational tool based on symbolic computation and optimization techniques that finds quadratizations for PDEs with a rational or polynomial right-hand side. We divide its construction into four algorithmic modules. 
The first module, described in Section~\ref{sec:mondecomp}, symbolically derives monomial decompositions from a given PDE system. In Section~\ref{sec:verif}, we introduce a second module that verifies whether a set of auxiliary variables is a quadratization. The third module, explained in Section~\ref{sec:search}, handles the search strategy to explore possible quadratizations based on the Branch-and-Bound framework \cite{B&B}. Lastly, the fourth module in Section~\ref{sec:ratPDE} presents a method for finding quadratizations for rational PDEs.

\subsubsection{Monomial decomposition of nonquadratic terms and Module 1}
\label{sec:mondecomp}
Given a polynomial symbolic PDE of the form~\eqref{eq:def-pde}, we first identify all possible ways of decomposing a nonquadratic monomial on the PDE's right-hand into a product of two variables, such that they are quadratic in the new set of auxiliary variables. These decompositions, which we denote as $\cM_1, \dots, \cM_m$, define the search space of potential quadratizations. 
    \begin{example}[Monomial decompositions]
    \label{ex:mondecomp}
        There are five possible decompositions that turn the quartic term $u_x^2u^2$ into a quadratic one: $\cM_1 := (1, u_x^2u^2)$, $\cM_2:=(u_x^2u, u)$, $\cM_3:=(u_x, u_xu^2)$, $\cM_4 :=(u_xu, u_xu)$ and $\cM_5 :=(u_x^2, u^2)$. Note that the expressions $1, u,$ and $u_x$ do not require the introduction of auxiliary variables. Then, if we choose decompositions $\cM_1$, $\cM_2$, or $\cM_3$ to rewrite the original monomial, we disregard the expressions $1, u,$ or $u_x$, respectively and define only one auxiliary variable: $u_x^2u^2$ for $\cM_1$, $u_x^2u$ for $\cM_2$, and $u_xu^2$ for $\cM_3$. Choosing $\cM_4$ would also require the definition of one auxiliary variable $u_xu$. If we choose decomposition $\cM_5$ instead, we must introduce two auxiliary variables: $u_x^2$ and $u^2$.
    \end{example}
Example~\ref{ex:mondecomp} shows that finding the decompositions of a monomial translates into a combinatorial problem, where the goal is to answer the question: In how many ways can we divide a set of elements of size $n$ into two subsets? We solve this problem to derive an upper bound on the number of possible monomial decompositions. Let the elements of a monomial be its polynomial variables repeated by their exponents. For example, in $u_x^2u^2$ the elements are $u_x, u_x, u, u$, where $u_x, u_x$, and $u, u$ are indistinguishable. We call $d$ the total number of such elements, that is, the total degree of the monomial; in the example $u_x^2u^2$, $d=4$. Now consider a monomial of degree $d$. For each element of this monomial, we decide whether to place it in the first of the two subsets into which we are dividing the original expression. This gives us $2^d$ possible options. Note that this result treats equivalent decompositions as distinct; for example, $(u_x^2,u^2)$ and $(u^2, u_x^2)$ are identified as different decompositions. Since we are only interested in combinations rather than ordered pairs, we divide the result $2^d$ by two: 
    \begin{equation}
    \label{eq:mondecompres}
        \frac{2^d}{2} = 2^{d-1}.
    \end{equation}
We assume that all elements within the monomial are distinguishable in our analysis. Consequently, \eqref{eq:mondecompres} provides an upper bound on the decompositions of a nonquadratic monomial. The conservativeness of this upper bound is shown in Example~\ref{ex:mondecomp}, where there are only five possible decompositions of $u_x^2u^2$, yet applying this bound gives $2^{d-1} = 2^3 = 8$ different tuples. 
    
To find decompositions $\cM_1, \dots, \cM_m$ from a PDE of the form~\eqref{def:pde-quad}, Module~\ref{alg:mon-decomp} first stores as tuples all possible ways to rewrite the lowest-degree nonquadratic monomial encountered on the right-hand side of the PDE as the product of two elements. We adopt this strategy motivated by the result in~\eqref{eq:mondecompres}, which implies that lower-degree monomials result in less decompositions. Next, Module~\ref{alg:mon-decomp} generates a set $\cD$ containing these decompositions and sorts them according to predefined heuristics designed to find a lower-order quadratization faster. We implement three heuristics that prioritize auxiliary variables with a higher chance of deriving lower-degree monomials after differentiation---namely, those with lower polynomial degree or lower order of derivatives. These are listed below.
\begin{itemize}
    \item[H1:] Sort by order and degree. Differentiating monomials with higher-order derivatives can lead to polynomials with many nonquadratic terms. To mitigate this, we sort the monomial decompositions by prioritizing expressions of lower order of derivatives and then of lower degree.
    \item[H2:] Sort by degree and order. The result in \eqref{eq:mondecompres} implies that decomposing a monomial with a high degree leads to a high number of decompositions, which translates into more candidates for quadratizations to explore. To reduce the number of decompositions, this rule sorts monomial decompositions by lower degree and secondly by lower order.
    \item[H3:] Sort by a function of order and degree. This rule sorts the monomial decompositions in descending order using a linear combination of the degree and order of the derivatives. Let $j$ be the highest order of derivatives and $d$ the total degree of a monomial; the default in \texttt{QuPDE} for this heuristic is $d + 2j$. 
\end{itemize}
In Module~\ref{alg:mon-decomp}, we describe the steps of this module through a pseudocode, while Example~\ref{ex:sort-heur} illustrates its application to a PDE. 

\begin{module}[H]
    \caption{(Monomial decomposition of nonquadratic terms)}
    \label{alg:mon-decomp}
        \begin{algorithmic}[1]
            \Require (i) Symbolically defined right-hand side of a PDE of the form~\eqref{eq:def-pde}, (ii) Sorting heuristic (H1, H2, or H3)
            \Ensure Set $\cD$ with tuples of size at most two sorted by the heuristic
            \State $p \gets$ lowest-degree nonquadratic monomial encountered on the PDE's right-hand side
            \State $\cM_1, \dots, \cM_m \gets$ monomial decompositions of $p$
            \State Remove from $\cM_1, \dots, \cM_m$ the monomials that do not propose auxiliary variables
            \State $\cD \gets$ $\{\cM_1, \dots, \cM_m\}$ 
            \State Sort $\cD$ according to heuristic \\
            \Return $\cD$
        \end{algorithmic}
    \end{module}

    \begin{example}[Application of Module~\ref{alg:mon-decomp}]
    \label{ex:sort-heur}
        Consider the one-dimensional quartic PDE
        \begin{equation}
            \label{eq:vardecomp}
                u_t = u^3u_x.
        \end{equation}
        Following step 1 of Module~\ref{alg:mon-decomp}, we identify $u^3u_x$ as the only nonquadratic monomial in the right-hand side of \eqref{eq:vardecomp}. In step 2, we compute the monomial decompositions of $u^3u_x$, from which we obtain the tuples $\cM_1:= (u, u^2u_x)$, $\cM_2:=(u^2, u_xu)$, $\cM_3 := (u^3, u_x)$ and $\cM_4:= (u_x^3u, 1).$ With each of these tuples, one can rewrite the original expression as quadratic.

        As stated in step 3, we dispose of monomials that do not propose auxiliary variables in each $\cM_i$---in this example, $u$, $u_x$, and 1---to arrive at the candidate auxiliary variables for quadratization listed in the set $\cD:=\{(u^2u_x), (u^2, u_xu), (u^3), (u_x^3u)\}$. Then, we follow step 4 and sort $\cD$ by a determined criterion. For illustration purposes, we show the result of sorting according to each of the heuristics defined above: using H1 we get the sorted $\cD =\{(u^3), (u^2, u_xu), (u^2u_x), (u_x^3u)\}$, using H2 we get $\cD =\{(u^2, u_xu), (u^3), (u^2u_x), (u_x^3u) \}$ and using H3 we get $\cD =\{(u^3), (u^2, u_xu), (u^2u_x), (u_x^3u)\}$.
    \end{example}
Having established how to obtain the auxiliary variables that lead to a quadratic form for a particular monomial, we next present the algorithmic module that determines whether a set of these variables is a quadratization.

\subsubsection{Verification of a quadratization and Module 2}
\label{sec:verif}
We now algorithmically verify whether a set of monomials (e.g., one produced by the application of Module~\ref{alg:mon-decomp}) forms a quadratization for a polynomial PDE. We follow the ideas in \cite{quadODE} and define three sets for this purpose. The first set, $\cV$, contains the polynomials of a potential quadratization, including the spatial derivatives of the auxiliary variables that are allowed according to the \text{\difforder} set for the quadratization. The second set, $\cV^2$, includes all polynomials that result from multiplying the elements of the first set with one another. The third set, $\mathcal{P}$, contains the polynomials on the right-hand side of the PDE and the time derivatives of the auxiliary variables.
    \begin{definition}
        (Definition of the sets $\mathcal{V}$, $\mathcal{V}^2$ and $\mathcal{P}$)
        \label{def:v-v2-pnq}
        Consider a vector
        \begin{equation*}
            \bw := \begin{bmatrix}
                w_1 & \dots & w_\ell
            \end{bmatrix}^\top, 
        \end{equation*}
        where $w_1, \dots, w_\ell\in\mathbb{R}[\bu, \partial_x\bu, \dots, \partial^m_x\bu]$ are auxiliary monomials that comprise a potential quadratization $\cW$ of \textit{\difforder} $k$ of a PDE of the form~\eqref{eq:def-pde}. 
        We define the sets $\mathcal{V}$, $\mathcal{V}^2$ and $\mathcal{P}$ as follows:
    \begin{itemize}
        \item $\mathcal{V}$ denotes the set of polynomials to be used for rewriting the PDE as quadratic. This set contains: the original variables of the PDE and their spatial derivatives up to order $k$, the auxiliary variables and their allowed spatial derivatives given by the set $\cC = \{ c_1, \dots, c_\ell \}$ and $k$, and the constant 1, i.e.,
        \begin{equation}
        \label{eq:V-set}
            \mathcal{V} \coloneq \left\{1, \bu, \partial_x\bu, \dots, \partial^{k}_x\bu, \partial_x^{\leq k-\cC}\bw \right\}.
        \end{equation} 
        \item $\mathcal{V}^2$ is the set of all polynomials that are generated by multiplying two elements of $\mathcal{V}$, i.e.,
        \begin{equation}
            \label{eq:V2-set}
            \mathcal{V}^2 \coloneq \{v_1 v_2 \hspace{1mm} | \hspace{1mm} v_1, v_2 \in \mathcal{V}\}.
        \end{equation}
        \item $\mathcal{P}$ is the set of the polynomials on the right-hand side of both the $n$ equations of the original PDE and the time derivatives corresponding to the $\ell$ auxiliary variables, i.e.,
        \begin{equation*}
            \mathcal{P}: =\{p_1(\bu, \partial_x\bu, \dots, \partial^{k}_x\bu), \dots, p_{n+\ell}(\bu, \partial_x\bu, \dots, \partial^{k}_x\bu)\}.       
            \end{equation*}
    \end{itemize}
    \end{definition}
Based on the definitions of $\mathcal{V}$, $\mathcal{V}^2$, and $\mathcal{P}$, it is sufficient to verify that the polynomials in $\mathcal{P}$ are linear combinations of elements in $\mathcal{V}^2$ to conclude that the corresponding set of auxiliary variables is a quadratization of the given PDE. This result is stated in the following proposition.

    \begin{proposition}
    \label{prop:verif}
        Consider the set $\cV^2$, and let $\cW$ be the set of auxiliary variables associated with $\mathcal{V}^2$. Then, $\cW$ is a quadratization of differential order $k$ if and only if all $p(\bu, \partial_x\bu, \dots, \partial^{k}_x\bu)\in\mathcal{P}$ belong to the linear span of $\cV^2$.
    \end{proposition}
    \begin{proof}
         Let us consider the potential quadratization $\cW$ and its corresponding vector $\bw$. Then, 
        \[\mathcal{V} = \left\{1, \bu, \partial_x\bu, \dots, \partial^{k}_x\bu, \partial_x^{\leq k-\cC}\bw \right\}.\]
        Let $p_j(\bu, \partial_x\bu, \dots, \partial^{k}_x\bu)\in\mathcal{P}$, $g_1(\bu, \partial_x\bu, \dots, \partial^{k}_x\bu), \dots, g_m(\bu, \partial_x\bu, \dots, \partial^{k}_x\bu)$ be all the polynomials in $\mathcal{V}^2$, and $\alpha_1, \dots \alpha_m\in\mathbb{R}$. First, we note from the definition of $\mathcal{V}$ and $\mathcal{V}^2$ that $g_1, \dots, g_m$ are polynomials resulting from the multiplication of at most two elements, where each element is either a variable in $\bu$, a spatial derivative of a component of $\bu$, an auxiliary variable, or a spatial derivative of an auxiliary variable. Then, if $p_j = \sum_{i=1}^n \alpha_i g_i$, it follows that $p_j$ can be written as a polynomial of degree at most two in the variables $\bu, \partial_x\bu, \dots, \partial^{k}_x\bu, \partial_x^{\leq k-\cC}\bw$. Note that $\mathcal{P}$ is the set of all polynomials of the right-hand side of the PDE system, including the time derivatives of the auxiliary variables. Then, if the condition above is satisfied for all $p_j\in\cP$ with $j\in\{1,\dots,n+\ell\}$, $\cW$ is a quadratization of the PDE. Conversely, if $\cW$ is a quadratization of the PDE, then all $p_j\in \cP$ are linear combinations of the elements in $\mathcal{V}^2$.
    \end{proof}
    To algorithmically check the condition in Proposition~\ref{prop:verif}, we define a vector $\bb \in \mathbb{R}^r[\bu, \partial_x\bu, \dots, \partial^{k}_x\bu]$ that contains all the monomials within $\mathcal{V}^2$ and $\mathcal{P}$. Namely, 
    \begin{equation*}
    \label{eq:vector-base}
    \bb \coloneq \begin{bmatrix}
        m_1 &
        m_2 &
        \dots &
        m_r
    \end{bmatrix}^\top,
\end{equation*}
    where $m_1, m_2, \dots, m_r \in \mathbb{R}[\bu, \partial_x\bu, \dots, \partial^{k}_x\bu]$ are all unique monomials within the polynomials in $\mathcal{V}^2$ and $\mathcal{P}$ ordered lexicographically. Moreover, we represent a polynomial as a vector $\bv \in \mathbb{R}^r$, where each nonzero entry corresponds to the presence of monomial $m_i$ (for $i=1, \dots, r$) in the polynomial ordered according to $\bb$, i.e.,
    \begin{equation}
        \label{eq:vector-pol}
        \mathbf{v} \coloneq (\alpha_i)_i,
    \end{equation}
    where $\alpha_i$ is the coefficient of $m_i$ in the polynomial. Next, we define a matrix $\bV^2\in\mathbb{R}^{p\times r}$ based on the polynomials in $\mathcal{V}^2$ as: 
    \begin{equation}
        \label{eq:pol-matrix}
        \bV^2 \coloneq (\alpha_i)_{ji},
    \end{equation}
where $\alpha_i$ is the coefficient of $m_i$ in the polynomial $g_j$, for $j=1,\dots, p$ with $p$ the size of the set $\mathcal{V}^2$. In this matrix, the rows are the transpose of each polynomial vector. Having this matrix representation allows us to apply Gaussian elimination to first obtain the reduced row-echelon form of $\bV^2$, which we denote by $\overline{\bV^2}$, and then use the rows of $\overline{\bV^2}$ to reduce each polynomial vector in $\mathcal{P}$. If the result of a reduction is not zero for a polynomial in $\mathcal{P}$, we store the nonquadratic reduced polynomial in a set that we denote by $\cR_{\rm{nq}}$. If $\cR_{\rm{nq}}$ remains empty after reducing all polynomials in $\mathcal{P}$, we return the quadratic transformation. Otherwise, we output the set $\cR_{\rm{nq}}$, which contains the nonquadratic polynomials from the reduction step that could not be made quadratic using $\cV^2$. We show the underlying algorithm of \texttt{QuPDE}'s verification module in Module~\ref{alg:verif-quad}.
     \begin{module}[H]
        \caption{(Verification of a quadratization)}
        \label{alg:verif-quad}
            \begin{algorithmic}[1]
                \Require (i) Symbolically defined polynomial PDE of the form~\eqref{eq:def-pde}, (ii) Set $\mathcal{W}$, (iii) Differential order $k$
                \Ensure Right-hand side polynomials of a quadratic PDE system of the form \eqref{eq:def-pde} or $\mathcal{R}_{\rm{nq}}$
                \State $\mathcal{P} \gets $ set of polynomials in the right-hand side of both the PDE and the differential equations of the elements in $\mathcal{W}$
                \State $\cV \gets$ set $\cV$ based on $\mathcal{W}$, $k$, and $\mathcal{P}$ as defined in~\eqref{eq:V-set} 
                \State $\bV^2 \gets$ matrix form of corresponding $\mathcal{V}^2$, defined in~\eqref{eq:V2-set}
                \State $\overline{\bV^2} \gets$ reduced row-echelon form of $\bV^2$
                \State $\cR_{\rm{nq}} \gets \{\}$
                \State $\bp_1, \dots, \bp_{n+\ell} \gets$ polynomials in $\cP$ in vector representation
                \For{$\bp_i \in \cP$}
                \State $\overline{\bp}_i \gets$ reduction of $\bp_i$ with respect to the rows of $\overline{\bV^2}$
                \If{$\overline{\bp}_i\neq0$} 
                \State Append $\overline{\bp}_i$ to $\mathcal{R_{\rm{nq}}}$
                \Else
                \State $g_i(\bu, \partial_x\bu, \dots, \partial^{k}_x\bu, \partial_x^{\leq k-\cC}\bw) \gets$ symbolic and quadratic representation of $\bp_i$  
                \EndIf
                \EndFor
            \If{$\mathcal{R_{\rm{nq}}}$ is empty}
            \Return $\{g_1, \dots, g_{n+\ell}\}$
            \Else{} \Return $\mathcal{R}_{\rm{nq}}$
            \EndIf
            \end{algorithmic}
        \end{module}
    For practical efficiency, the software implementation of Module~\ref{alg:verif-quad} performs all reductions at the polynomial level. The explicit matrix representation presented in this section serves illustrative purposes. In the following example, we show how Module~\ref{alg:verif-quad} verifies if a set of auxiliary variables is a quadratization for a particular PDE. 
    \begin{example}[Verification of a quadratization]
    \label{ex:verif}
        Consider the cubic PDE
        \begin{equation}
        \label{eq:ex-verif}
            u_t=u_xu^2.
        \end{equation}
         We introduce the auxiliary variable $w:=u^2$ and write the differential equation $w_t = 2u_tu=2u_xu^3$. Additionally, we aim to obtain a quadratization of \textit{\difforder} 1, so we calculate the first-order spatial derivative of the auxiliary variable, $w_x=2u_xu$. To apply Module~\ref{alg:verif-quad} to verify whether the set $\cW := \{w\}$ is a quadratization for~\eqref{eq:ex-verif}, we set $\mathcal{P}=\{u_t, w_t\}=\{u_xu^2, 2u_xu^3\}$ and construct $\mathcal{V}$ by listing the polynomials that are part of the potential quadratization for~\eqref{eq:ex-verif}, i.e., $\mathcal{V} := \{1, u, u_x, u^2, 2u_xu\}$. It follows \[\mathcal{V}^2 = \{1, u, u_x, u^2, 2u_xu, u_xu, u^3, 2u_xu^2, u_x^2, u^2u_x, 2u_x^2u, u^4, 2u_xu^3, 4u_x^2u^2\}.\] We set the monomial basis $\bb$ as
        \begin{equation}
            \label{eq:vector-rep-ex}
            \bb=\begin{bmatrix}
                1 & u & u_x & u^2 & u_x^2 & u_xu & u^3 & u_xu^2 & u_x^2u & u^4 & u_xu^3 & u_x^2u^2 
            \end{bmatrix}^\top.
        \end{equation}
        Then, the matrix representation of $\mathcal{V}^2$ is
        \begin{equation*}
        \label{eq:matrix-rep-ex}
            \bV^2 =\begin{bmatrix}
                1 & 0 & 0 & 0 & 0 & 0 & 0 & 0 & 0 & 0 & 0 & 0\\
                0 & 1 & 0 & 0 & 0 & 0 & 0 & 0 & 0 & 0 & 0 & 0\\
                0 & 0 & 1 & 0 & 0 & 0 & 0 & 0 & 0 & 0 & 0 & 0\\
                0 & 0 & 0 & 1 & 0 & 0 & 0 & 0 & 0 & 0 & 0 & 0\\
                0 & 0 & 0 & 0 & 0 & 2 & 0 & 0 & 0 & 0 & 0 & 0\\
                0 & 0 & 0 & 0 & 0 & 1 & 0 & 0 & 0 & 0 & 0 & 0\\
                0 & 0 & 0 & 0 & 0 & 0 & 1 & 0 & 0 & 0 & 0 & 0\\
                0 & 0 & 0 & 0 & 0 & 0 & 0 & 2 & 0 & 0 & 0 & 0\\
                0 & 0 & 0 & 0 & 1 & 0 & 0 & 0 & 0 & 0 & 0 & 0\\
                0 & 0 & 0 & 0 & 0 & 0 & 0 & 1 & 0 & 0 & 0 & 0\\
                0 & 0 & 0 & 0 & 0 & 0 & 0 & 0 & 2 & 0 & 0 & 0\\
                0 & 0 & 0 & 0 & 0 & 0 & 0 & 0 & 0 & 1 & 0 & 0\\
                0 & 0 & 0 & 0 & 0 & 0 & 0 & 0 & 0 & 0 & 2 & 0 \\
                0 & 0 & 0 & 0 & 0 & 0 & 0 & 0 & 0 & 0 & 0 & 4 \\
            \end{bmatrix}.
        \end{equation*}
        Following step 2, we apply Gaussian elimination to $\bV^2$ and obtain:
        \begin{equation*}
            \label{eq:matrix-rep-ex-2}
            \overline{\bV^2} := \begin{bmatrix}
                \bI_{12\times12}\\
                \mathbf{0}_{2\times12}
            \end{bmatrix},
        \end{equation*}
        which yields the set of polynomials: $$\{1, u, u_x, u^2, u_x^2, u_xu, u^3, u_xu^2, u_x^2u, u^4, u_xu^3, u_x^2u^2, 0, 0\}.$$

        Let $\overline{v_1}, \dots, \overline{v_{12}}$ the polynomials represented by the nonzero row vectors in $\overline{\bV^2}$. Then, the reduction of the nonquadratic expressions stored in $\mathcal{P}$ is
        \begin{align*}
            \overline{u_t} &= u_t - \overline{v}_{8} = u_xu^2 - u_xu^2 = 0, \\
            \overline{w_t} &= w_t - 2\cdot \overline{v}_{11} = 2u_xu^3 - 2\cdot u_xu^3 = 0.
        \end{align*}
        Since $\overline{u_t}$ and $\overline{w_t}$ are both equal to zero, $\mathcal{\cR_{\rm{nq}}} = \varnothing$, which confirms that $\{u^2\}$ is a quadratization for the PDE~\eqref{eq:ex-verif}.
    \end{example}
    
In Section~\ref{sec:mondecomp}, we have shown that there are numerous ways to define auxiliary variables to express a term as quadratic. Here, in Section~\ref{sec:verif}, we have introduced a mechanism to identify whether a particular set of auxiliary variables is a quadratization. Next, in Section~\ref{sec:search}, we describe an optimization algorithm that determines the order in which Module~\ref{alg:verif-quad} verifies whether the sets produced by Module~\ref{alg:mon-decomp} are quadratizations, and yields the best one among those found.

\subsubsection{Efficient search strategies and Module 3}
\label{sec:search}
        We implement a search algorithm using a \textit{Branch-and-Bound} (B\&B) \cite{B&B} approach to find a low-order quadratization of a PDE, along with additional strategies to speed up the search. The space of elements traversed by this algorithm, which we call the \textit{search space}, consists of all sets of auxiliary variables introduced by Module~\ref{alg:mon-decomp}. We assign to this search space a tree representation, which we call the \textit{search tree}, as illustrated in Figure~\ref{fig:hierarchy-tree}. Here, the root is an empty set and its children nodes are the monomial decompositions generated by Module~\ref{alg:mon-decomp}; the level of a node in the tree is determined by the number of branching steps between that node and the root; and each node represents a \textit{subproblem}, which has a corresponding PDE that results from adding to the original PDE the auxiliary variables in the current node and its ancestors. 
        \begin{figure}[H]
            \centering
            \includegraphics[scale=0.3]{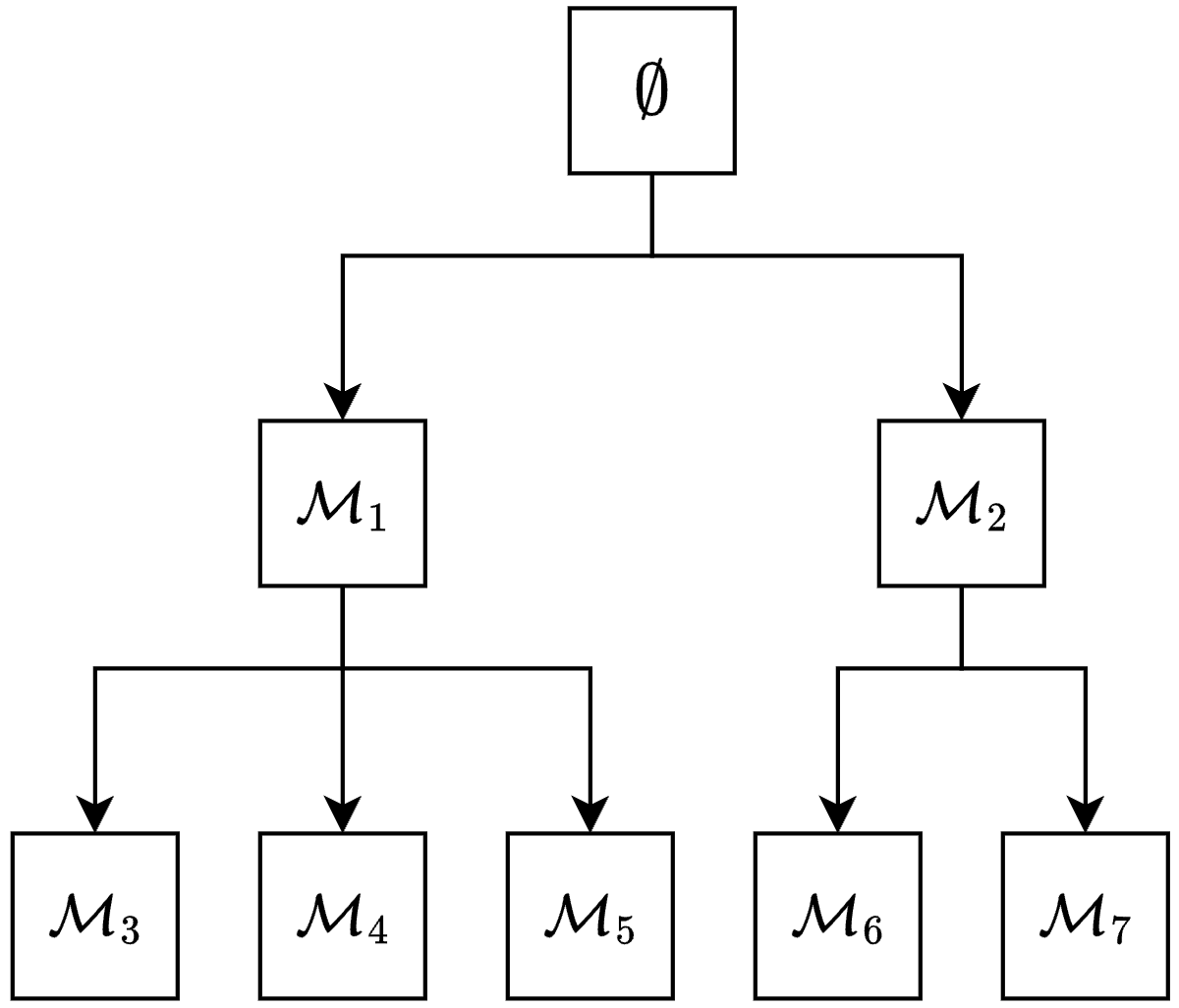}
            \caption{Example of a search tree of monomial decompositions represented by $\cM_1, \dots, \cM_7$.}
            \label{fig:hierarchy-tree}
        \end{figure}
        We formally state how the problem of finding a low-order quadratization through the search tree translates to the B\&B framework using the standard B\&B terminology (e.g., \cite{B&B}). 
        \begin{problem}(Branch-and-Bound formulation for finding a low-order quadratization)
        \label{prob:bnb-form}
                \begin{itemize}
                    \item The search space $\cS$ is the set of all monomials in the variables of a PDE of the form~\eqref{eq:def-pde} and their spatial derivatives up to order $k$. Namely, 
                    \begin{align}
                    \label{eq:set_mon}
                        \cS \subseteq \{m\in \real[\bu, \partial_x\bu, \dots, \partial^{k}_x\bu]\hspace{1mm} | \hspace{1mm} m \text{ is a monomial}\}.
                    \end{align}
                    \item The objective is to minimize the number of monomials introduced in a quadratization. This leads to the problem of finding the quadratization 
                    \begin{align*}
                        \cW^* =  \argmin_{\cW\subseteq\cS}|\cW|,
                    \end{align*}
                    where $\cW$ is a quadratization of \eqref{eq:def-pde} and $|\cdot|$ denotes the cardinality of a set.
                    \item Each subproblem is defined by a finite subset of monomials $\cM\subseteq\cS$, such as those produced by the monomial decompositions of a nonquadratic term within the set $\mathcal{R}_{\rm{nq}}$, introduced in Module~\ref{alg:verif-quad}.
                \end{itemize}
            
        \end{problem}
        The B\&B framework requires three additional components: the search strategy that defines the order in which subproblems in the tree are explored; the branching strategy, which partitions the search space to produce new subproblems in the tree; and the pruning rules that prevent the exploration of regions in the search tree that cannot yield better solutions. Below, we describe the B\&B components of the search and branch strategy for the quadratization problem along with the pruning rules implemented.
            \paragraph{Search strategy} We use a depth-first search with backtracking \cite{TARJANDEPTHFIRST, B&B}. Namely, the algorithm explores the search tree as far as possible along each node until it finds a quadratization, and then it retraces to unexplored subproblems.
             \paragraph{Branching strategy} We apply a wide branching technique, where the children nodes of a subproblem are generated by applying Module~\ref{alg:mon-decomp} to the nonquadratic polynomials of its associated PDE---elements of $\mathcal{R}_{\rm{nq}}$. Note that at the root of the tree, $\mathcal{R}_{\rm{nq}}$ simply contains the nonquadratic terms of the right-hand side of the original PDE. Then, the children nodes represented by potential auxiliary variables are organized according to the sorting heuristic chosen from those introduced in Section~\ref{sec:mondecomp} (H1, H2, H3). This sorting determines the order in which the subproblems are explored.
             \paragraph{Pruning rules} We implement two pruning rules to prevent the exploration of suboptimal regions of the search tree. The first rule is essential for the algorithm, as it directly targets the objective function to minimize---the number of auxiliary variables introduced. The second rule is related to the order of partial derivatives we allow within the quadratic equations and hence to the regularity of the PDE solutions. We detail these rules below:
            \begin{enumerate}
                \item[PR1:] The smallest order---that is, the minimum number of monomials---of a quadratization found within the visited nodes. In the first run of the algorithm, this value is initialized at a predefined bound that the user can change. This bound determines the maximum depth of the tree.

                \item[PR2:] The order of derivatives within the auxiliary variables. This pruning rule enforces user-specified regularity restrictions on the existence and continuity of higher-order spatial derivatives of the solution, and guarantees that the order of derivatives within the transformed quadratic system does not exceed the \textnormal{\difforder} set for the quadratization. Therefore, if a branch introduces monomials with higher-order derivatives than those allowed for the problem, the algorithm prunes that branch. Motivated by the result in Theorem~\ref{thm:existence}, \texttt{QuPDE}'s default \textnormal{\difforder} for a quadratization is three times the highest order of the system ($3h$). Since auxiliary variables are differentiated with respect to space $h$ times to obtain their differential equations, by default, the algorithm prunes branches that introduce variables with spatial derivatives of order higher than $2h$, i.e., $3h-h=2h$. Nonetheless, the user can modify this upper bound.
        \end{enumerate}
     
     Next, we explain the steps of the B\&B search algorithm implemented in Module~\ref{alg:b&b}. The search algorithm starts at the empty root node, as depicted in Figure~\ref{fig:hierarchy-tree}, with a predetermined bound for the largest quadratization it will admit. Then, at each branch, the algorithm evaluates whether it should be pruned based on the defined pruning rules. If Module~\ref{alg:b&b} confirms that a particular branch does not produce a quadratization of lower order than the best one found so far, Module~\ref{alg:verif-quad} checks whether the corresponding set of auxiliary variables in that node is a quadratization. If so, the algorithm returns the quadratization, updates the best solution found so far, and does backtracking. If the auxiliary variables do not yield a quadratization, the search algorithm generates the children of the current node, i.e., its subproblems, and explores them using a depth-first search strategy. After finishing the search throughout the tree, the algorithm returns the smallest set of auxiliary variables, which is a low-order quadratization for the PDE. 
     Module~\ref{alg:b&b} shows the instructions of this algorithm through a pseudocode, and Example~\ref{ex:b&b} illustrates how this module works for a given PDE. 
        
        \begin{module}[H]
        \caption{(Search algorithm using Branch-and-Bound framework)}
        \label{alg:b&b}
            \begin{algorithmic}[1]
                \Require (i) Symbolically defined PDE system of the form~\eqref{eq:def-pde}, (ii) Set of symbolic auxiliary monomials $\cW\subseteq\cS$, (iii) Best quadratization found so far $\widehat{\cW}\subseteq\cS$ of size $\ell$ (initialized at a predetermined bound)
                \Ensure Quadratization $\cW^*\subseteq\cS$
                \If{any of the pruning rules applied to $\cW$ outputs \texttt{True}}
                \Return $\widehat{\cW}$
                \EndIf
                \If{Module~\ref{alg:verif-quad} verifies that $\cW$ is a quadratization for the PDE}
                    \If{$|\cW|<|\widehat\cW|$}
                    \Return $\cW$
                    \Else{} \Return $\widehat{\cW}$
                \EndIf
                \Else
                    \State $\mathcal{R}_{\rm{nq}} \gets $ result of applying Module~\ref{alg:verif-quad} to the PDE with the set $\cW$ 
                    \State $\cD \gets$ monomial decompositions generated by Module~\ref{alg:mon-decomp} of $\mathcal{R}_{\rm{nq}}$
                    \For{$\cM_i\in \cD$}
                        \State $ \cW \gets \cW\cup \cM_i$
                        \State $\overline{\cW}\gets $ call Module~\ref{alg:b&b} with the inputs: 
                        (i) the new PDE system consisting of the original equations plus the differential equations from $\cW$, (ii) $\cW$ as the current set of auxiliary variables, and (iii) $\widehat{\cW}$ as the best solution found so far
                        \If{$|\overline{\cW}|<|\widehat\cW|$} $\widehat\cW \gets \overline{\cW}$
                        \EndIf
                        \State Remove $\cM_i$ from $\cW$\vspace{4px}
                    \EndFor
                    \Return $\cW^*\gets \widehat\cW$
                \EndIf
            \end{algorithmic}
        \end{module} 
    
        \begin{example}[B\&B search algorithm]
        \label{ex:b&b}
            Consider the cubic PDE
                \begin{equation}
                \label{eq:ex-b&b}
                    u_t = u^2u_{xxx}.
                \end{equation}
            We decompose the nonquadratic monomial $u^2u_{xxx}$ into $\cM_1 := (u_{xxx}u^2, 1)$, $\cM_2:=(u_{xxx}u, u)$ or $\cM_3:=(u^2, u_{xxx})$, remove the already existing expressions $1$, $u$, and $u_{xxx}$, and sort the resulting tuples by the heuristic H3 to obtain the set $\cD := \{(u^2), (u_{xxx}u), (u_{xxx}u^2)\}$. Then, based on $\cD$, we draw the search tree as
            \begin{align*}
            \scalebox{0.72}{ 
                \begin{tikzpicture}[
                  level distance=1.3cm,
                  level 1/.style={sibling distance=2.7cm},
                  every node/.style={font=\Large},
                  edge from parent/.style={draw, -, thick},
                  edge from parent path={(\tikzparentnode) -- (\tikzchildnode)},
                  baseline=(root.base)
                ]
                  \node (root) {$\varnothing$}
                    child {node {$(u^2)$}}
                    child {node {$(u_{xxx}u)$}}
                    child {node {$(u_{xxx}u^2)$}};
                \end{tikzpicture}
            }
            \end{align*}
           where each node in parentheses represents the variables we are adding to the system at that tree level. 
           
           Now the algorithm enters the first branch, i.e., the node described by $(u^2)$. The subproblem that this node represents is the PDE system we obtain by introducing the auxiliary variable in the set $\{u^2\}$ to the original PDE. We define $w:=u^2$ and write the new PDE: 
            \begin{align}
                \label{eq:ex-b&b2}
                u_t &= u^2u_{xxx}\\
                \label{eq:ex-b&b3}
                w_{t} &= 2u^3u_{xxx}.
            \end{align}
            After applying Module~\ref{alg:verif-quad} to~\eqref{eq:ex-b&b2} and~\eqref{eq:ex-b&b3}, we obtain that the reduced $w_{t}$ with respect to the set $\cV^2$ is $\overline{w_t} =w_t - w_{xxx}w + \frac{3}{2}w_xw_{xx} =6u_x^3u $. This results in $\mathcal{R}_{\rm{nq}}=\{6u_x^3u\}$, confirming that the set $\{u^2\}$ is not a quadratization. Therefore, we explore the subproblems of this node by executing Module~\ref{alg:mon-decomp} for $\mathcal{R}_{\rm{nq}}$, which in this case computes the decompositions of the nonquadratic monomial $u_x^3u$. After sorting these decompositions by H3, the tree for this search takes the form 
             \begin{align*}
            \hspace{-3.1cm}
                \scalebox{0.72}{ 
                \begin{tikzpicture}[
                  level distance=1.2cm,
                  level 1/.style={sibling distance=2.8cm},
                  level 2/.style={sibling distance=2.5cm},
                  every node/.style={font=\Large},
                  edge from parent/.style={draw, -, thick},
                  edge from parent path={(\tikzparentnode) -- (\tikzchildnode)},
                  baseline=(root.base)
                ]
                \node (root) {$\varnothing$}
                  child {node {$(u^2)$}
                    child {node {$(u_x^2, u_xu)$}}
                    child {node {$(u_x^2u)$}}
                    child {node {$(u_x^3)$}}
                    child {node {$(u_x^3u)$}}
                  }
                  child {node {$(u_{xxx}u)$}}
                  child {node {$(u_{xxx}u^2)$}
                  };
                \end{tikzpicture}
                }
            \end{align*}
            
            When the algorithm visits the node represented by $(u_x^2, u_xu)$, it executes Module~\ref{alg:verif-quad} on $\{u^2, u_x^2, u_xu\}$ and confirms that this set is a quadratization. This is depicted with the green box in the tree: 
            \begin{equation}
            \resizebox{1.0\textwidth}{!}{
                \label{eq:fig3-b&bex}
                \begin{forest} 
                for tree={
                    s sep=-1.7mm,
                    l sep=10mm,
                    font=\large,                    
                }
                    [{\Large$\varnothing$}, 
                    [$(u^2)$, 
                    for level=1{   
                      l sep=20mm
                    }, edge label={node[midway, fill=white, inner sep=1pt, text=blue, ]{1}}
                    [{$(u_x^2, u_xu)$}, tikz={\node [draw,green!45!black,inner sep=0,fit to=tree]{};}, edge label={node[midway, fill=white, inner sep=1pt, text=blue,  ]{2}}] 
                    [{$(u_x^2u)$}, edge label={node[midway, fill=white, inner sep=1pt, text=blue, ]{3}}] 
                    [{$(u_x^3)$}, edge label={node[midway, fill=white, inner sep=1pt, text=blue, ]{4}}]
                    [{$(u_x^3u)$}, edge label={node[midway, fill=white, inner sep=1pt, text=blue, ]{5}}]
                    ] 
                    [{$(u_{xxx}u)$},
                    for level=1{  
                      l sep=20mm
                    }, edge label={node[midway, fill=white, inner sep=1pt, text=blue, ]{6}}
                    [{$(u_x^2u)$}, edge label={node[midway, fill=white, inner sep=1pt, text=blue,]{7}}]
                    [{$( u_x^2, u_{xxxx}u )$}, text opacity=0.5, edge={opacity=0.5}]
                    [{$( u_xu_{xxxx}, u_xu)$}, text opacity=0.5, edge={opacity=0.5}]
                    [{$(u_x^2u_{xxxx})$}, edge label={node[midway, fill=white, inner sep=1pt, text=blue, ]{8}}]
                    [{$(u_xu_{xxxx}u)$}, edge label={node[midway, fill=white, inner sep=1pt, text=blue, ]{9}}]
                    [{$(u_x^2u_{xxxx}u)$}, edge label={node[midway, fill=white, inner sep=1pt, text=blue, ]{10}}]]
                    [{$(u^2u_{xxx})$},
                    for level=1{  
                      l sep=20mm
                    }, edge label={node[midway, fill=white, inner sep=1pt, text=blue, ]{11}}
                    [{$(u_{xxx}u)$}, edge label={node[midway, fill=white, inner sep=1pt, text=blue, ]{12}}]
                    [{$(u_{xxx}^2, u^3)$}, text opacity=0.5, edge={opacity=0.5}]
                    [{$(u_{xxx}^2u, u^2)$}, text opacity=0.5, edge={opacity=0.5}]
                    [{$(u^3u_{xxx})$}, edge label={node[midway, fill=white, inner sep=1pt, text=blue, ]{13}}]
                    [{$(u^2u_{xxx}^2)$}, edge label={node[midway, fill=white, inner sep=1pt, text=blue, ]{14}}]
                    [{$(u^3u_{xxx}^2)$}, edge label={node[midway, fill=white, inner sep=1pt, text=blue, font=]{15}}]]
                    ]]
                \end{forest}}
                \end{equation}
            where the order of traversed branches is labeled with blue numbers. After finding the quadratization on the node in the green box, it backtracks to the parent node $(u^2)$ and starts exploring nodes with possible lower-order quadratizations. From the third to the fifteenth step in~\eqref{eq:fig3-b&bex},
            the algorithm verifies that the nodes given by those branches are not quadratizations. Moreover, since the order of the best quadratization found so far is three, it prunes the four branches depicted in gray. 
            After verifying in step fifteen that the last lower-order quadratization candidate, the set $\{u^2u_{xxx}, u^3u_{xxx}^2\}$, is not a quadratization for~\eqref{eq:ex-b&b}, Module~\ref{alg:b&b} yields $\{u^2, u_x^2, u_xu\}$ as the best solution.
        \end{example}
    We can further reduce the number of nodes traversed by Module~\ref{alg:b&b} and increase its chances of finding a smaller quadratization by introducing additional steps to the search strategy. This improvement is demonstrated in Example~\ref{ex:shrink-ex} below, where finding a lower-order quadratization for PDE~\eqref{eq:ex-b&b} requires visiting only four nodes instead of fifteen.

    \subsubsection*{Improvements to the search strategy}
    
    We introduce an additional step to find a bound to prune more branches within the search tree and have a better chance of finding lower-order quadratizations. Given a polynomial PDE, if a set of $\ell$ auxiliary variables yields a quadratization, it happens in practice that a smaller subset could also be a quadratization for the PDE. In our implementation, we added an extra step where each time Module~\ref{alg:b&b} finds a quadratization, it tries to find a better one within all the subsets of the quadratization found. If Module~\ref{alg:b&b} finds a smaller quadratization through this extra step, it uses this new bound for pruning future branches. We illustrate the effect of this modification by revisiting Example~\ref{ex:b&b}.
            \begin{example}[Modification to the B\&B search algorithm]
            \label{ex:shrink-ex}
                We consider the PDE~\eqref{eq:ex-b&b}. When Module~\ref{alg:b&b} finds the first quadratization of this PDE in its second step depicted in~\eqref{eq:fig3-b&bex}, the state of the search tree is:
                \begin{align*}
                    \hspace{-3.1cm}
                    \scalebox{0.73}{ 
                    \begin{tikzpicture}[
                      level distance=1.2cm,
                      level 1/.style={sibling distance=2.8cm},
                      level 2/.style={sibling distance=2.5cm},
                      every node/.style={font=\Large},
                      edge from parent/.style={draw, -, thick},
                      edge from parent path={(\tikzparentnode) -- (\tikzchildnode)}
                    ]           
                    \node (root) {$\varnothing$}
                      child {node (n1) {$(u^2)$}
                        child {node (a1) {$(u_x^2, u_x u)$}
                        edge from parent node[midway, blue, font=\small,fill=white,inner sep=1pt] {2}
                        }
                        child {node (a2) {$(u_x^2 u)$}}
                        child {node (a3) {$(u_x^3)$}}
                        child {node (a4) {$(u_x^3 u)$}}
                        edge from parent node[midway,blue, font=\small,fill=white,inner sep=1pt] {1}
                      }
                      child {node (n2) {$(u_{xxx}u)$}}
                      child {node (n3) {$(u_{xxx}u^2)$}};
                    \node[draw=green!50!black, thick, fit=(a1), inner sep=0.5pt] {};
                    \end{tikzpicture}
                    }
                \end{align*}
                Instead of searching for smaller quadratizations in other branches, Module~\ref{alg:b&b} calls Module~\ref{alg:verif-quad} for each subset of the quadratization found: $\{u^2,u_x^2, u_xu\}$. This is depicted below
                
                \begin{align}
                    \label{eq:fig2-shrinkex}
                    \hspace{-5cm}
                    \vcenter{\hbox{
                    \scalebox{0.73}{ 
                    \begin{tikzpicture}[
                      level distance=1.2cm,
                      level 1/.style={sibling distance=2.8cm},
                      level 2/.style={sibling distance=2.5cm},
                      level 3/.style={sibling distance=1.8cm, level distance=2cm},
                      every node/.style={font=\Large},
                      edge from parent/.style={draw, -, thick},
                      edge from parent path={(\tikzparentnode) -- (\tikzchildnode)}
                    ]           
                    \node (root) {$\varnothing$}
                      child {node (n1) {$(u^2)$}
                        child {node (a1) {$( u_x u, u_x^2)$} 
                            child[dotted] {node (c1) {$\{u^2\}$}}
                            child[dotted] {node (c2) {$\{u_x^2\}$}}
                            child[dotted] {node (c3){$\{u_xu\}$}}
                            child[dotted] {node (c4) {$\{u^2, u_x^2\}$}}
                            child[dotted] {node {$\dots$}}
                            edge from parent node[midway, blue, font=\small,fill=white,inner sep=1pt] {2}
                            }
                        child[opacity=0.4] {node (a2) {$(u_x^2 u)$}}
                        child[opacity=0.4] {node (a3) {$(u_x^3)$}}
                        child[opacity=0.4] {node (a4) {$(u_x^3 u)$}}
                        edge from parent node[midway, blue, font=\small,fill=white,inner sep=1pt] {1}
                      }
                      child {node (n2) {$(u_{xxx}u)$}
                      edge from parent node[midway, blue, font=\small,fill=white,inner sep=1pt] {3}
                      }
                      child {node (n3) {$(u_{xxx}u^2)$}
                      edge from parent node[midway, blue, font=\small,fill=white,inner sep=1pt] {4}
                      };
                    \node[draw=green!50!black, thick, fit=(a1), inner sep=0.5pt] {};
                    \node[draw=green!50!black, thick,  fit=(c4), inner sep=0.5pt] {};
                    \end{tikzpicture}
                    }
                    }}
                \end{align}
                where the dotted line arrows distinguish these subsets from the formal search tree nodes. As shown in~\eqref{eq:fig2-shrinkex}, Module~\ref{alg:b&b} searches for a quadratization among the subsets $\{u^2\}, \{u_x^2\}, \{u_xu\}, \{u^2, u_x^2\}$. In the last set of variables, it identifies a quadratization; then, it stops searching within the rest of the node's subsets and updates the minimum quadratization order found. Given the update of the pruning rule, the path that Module~\ref{alg:b&b} follows through the tree is the one illustrated in~\eqref{eq:fig2-shrinkex} in blue numbers, where the resulting best quadratization is $\{u^2, u_x^2\}$.
            \end{example}
            
Having described the complete pipeline for obtaining quadratizations of polynomial PDEs, we introduce a method for handling rational PDEs, building on the algorithms presented in Module~\ref{alg:mon-decomp}, Module~\ref{alg:verif-quad}, and Module~\ref{alg:b&b}.

\subsubsection{Polynomialization of rational PDEs and Module 4}
\label{sec:ratPDE}

Nonlinear PDEs can have rational functions in their drift. Some examples are the Euler equations \cite{EULEREQ} and the HUX (Heliospheric Upwinding eXtrapolation) model that predicts the heliospheric solar wind speed \cite{Riley2011, solarwind_2023}. However, the modules presented so far only handle PDEs of polynomial nature. One way to find a quadratization of a PDE with rational terms using our tools is to transform the right-hand side expressions into polynomials by introducing rational auxiliary variables, a technique commonly referred to as \textit{polynomialization}. For example, consider the PDE $u_t = \frac{1}{u+1}$. We could introduce an auxiliary variable $q=\frac{1}{u+1}$, and the system would have the differential equations: $u_t = q$ and $q_t=-\frac{1}{(u+1)^3}=-q^3$. The problem with this approach lies in the omission of the relation $q=\frac{1}{u+1}$ by the algorithmic modules, which would mean disregarding the constraint $(u+1)\cdot q = 1$ during the quadratization procedure. This problem affects our solution's ability to find low-order quadratizations for the original PDE as it may not recognize that a polynomial is quadratic if we apply the simplification $(u+1)\cdot q = 1$, as it is shown in Example~\ref{ex:rat-module}.
    
Consider $p_1,\dots, p_n\in \real[\bu, \partial_x\bu, \dots, \partial^h_x\bu]$ and $d_1,\dots, d_n\in \real[\bu, \partial_x\bu, \dots, \partial^h_x\bu]$. We present Module~\ref{alg:rat-pdes}, which addresses PDEs of the form 
    \begin{align}
        \label{eq:rat-pde}
        \partial_t u_1=\dfrac{p_1(\bu, \partial_x\bu, \dots, \partial^h_x\bu)}{d_1(\bu, \partial_x\bu, \dots, \partial^h_x\bu)}, \dots, \partial_t u_n=\dfrac{p_n(\bu, \partial_x\bu, \dots, \partial^h_x\bu)}{d_n(\bu, \partial_x\bu, \dots, \partial^h_x\bu)},
    \end{align}
and includes the rational constraints while finding a quadratization. This module finds the partial fraction decomposition of a multivariate rational function by following \cite[Algorithm~2]{MultivariateApart2021}. There, the authors propose a two-step procedure given a multivariate rational fraction: (i) construct a special set of polynomials called a Gr{\"o}bner basis \cite{groebner} using a set of polynomials known as an ideal \cite{Cox2015}, and (ii) reduce the rational function with respect to this basis. To illustrate this approach, we introduce the definitions of a polynomial ideal and a Gr{\"o}bner basis. In polynomial algebra, a set of polynomials $g_1,\dots,g_m\in\mathbb{F}[\upsilon_1, \dots, \upsilon_n]$, called generators, define an ideal set denoted by $\mathcal{I}$. This ideal consists of all possible polynomial combinations of the generators, where the coefficients are the polynomials themselves: $\mathcal{I}:= \langle g_1,\dots,g_m \rangle =\{\sum_i f_i g_i \hspace{1mm}| \hspace{1mm}f_i\in\mathbb{F}[\upsilon_1, \dots, \upsilon_n]\}$, where the notation $\langle \cdot \rangle$ indicates that $\cI$ is generated by $g_1,\dots,g_m$. A Gr{\"o}bner basis is a set of generators of a particular ideal that allows one to find a unique remainder when reducing other polynomials with respect to this basis.

To apply these ideas in Module~\ref{alg:rat-pdes}, we consider $f_1, \dots, f_m$ as the irreducible factors of the denominators of a rational function within~\eqref{eq:rat-pde}, where $f_1, \dots, f_m \in \mathbb{R}[\bu, \partial_x\bu, \dots, \partial^h_x\bu]$. Now, let $q_1,\dots,q_m$ be the formal variables representing the inverses of $f_1, \dots, f_m$, e.g., $q_i=\frac{1}{f_i}$ for $i\in\{1,\dots, m\}$, and consider the ideal 
        \begin{equation}
        \label{eq:ideal}
        \cI = \langle  q_1 f_1 - 1, \dots, q_m f_m - 1 \rangle,
        \end{equation}
where $\cI \subset \mathbb{R}[q_1, \dots, q_m, \bu, \partial_x\bu, \dots, \partial^h_x\bu]$. Setting the generators $q_i f_i - 1=0$ imposes the constraint $q_i = 1 / f_i$. Motivated by this, we construct a Gr{\"o}bner basis corresponding to the generators of the ideal~\eqref{eq:ideal} and use it to reduce polynomials during quadratization steps involving polynomial operations, e.g., Gaussian elimination in Module~\ref{alg:verif-quad}. This ensures that any product of the form $q_if_i$ will be replaced by 1. The detailed steps of this procedure are listed in pseudocode in Module~\ref{alg:rat-pdes}, and Example~\ref{ex:rat-module} illustrates the integration of Module~\ref{alg:rat-pdes} into the quadratization procedure.
    
    \begin{module}[H]
        \caption{(\texttt{QuPDE}'s module for PDEs with rational terms)}
        \label{alg:rat-pdes}
        \begin{algorithmic}[1]
            \Require Symbolically defined rational PDE system of the form~\eqref{eq:rat-pde}
            \Ensure (i) Set of auxiliary variables $\cW$, (ii) A Gr{\"o}bner basis corresponding to the constraints of the rational variables, (iii) The PDE system in polynomial form expressed in terms of the rational auxiliary variables 
            \State $\cW \gets \{\}$
            \For{each right-hand side expression in the PDE system}
            \State Rewrite the expression as $\dfrac{p(\bu, \partial_x\bu, \dots, \partial^h_x\bu)}{d(\bu, \partial_x\bu, \dots, \partial^h_x\bu)}$
            \State Cancel common factors between $p$ and $d$
            \State $\{f_1, \dots, f_m\} \gets$ terms of the factorization $d=\prod_{i=1}^mf_i^{j_i}$, with $j_i>0$
            \For{$f_i\in \{f_1, \dots, f_m\}$}
            \State $q_i \gets \frac{1}{f_i}$
            \If{$q_i \notin \cW $}
            \State $\mathcal{I} \gets \mathcal{I}\cup\{q_if_i - 1\}$
            \State $\cW \gets \cW\cup\{q_i\}$
            \EndIf
            \EndFor            
            \EndFor
            \State $\cG \leftarrow $ Gr{\"o}bner basis of $\mathcal{I}$
            \State Rewrite each right-hand side of the rational PDE system using the variables in $\cW$, add to the system the differential equations corresponding to the auxiliary variables in $\cW$, and reduce each expression with respect to $\cG$ to obtain a polynomial system\\
            \Return $\cW$, $\cG$ and the polynomial PDE system 
            
        \end{algorithmic}
    \end{module}

    \begin{example}[Quadratization for a rational PDE]
    \label{ex:rat-module}
        Consider the PDE
        \begin{align}
            \label{eq:ex-rat-1}
            & u_t = \dfrac{u_{xxx}}{u}.
        \end{align}
    Since \eqref{eq:ex-rat-1} has a rational term, we follow the steps in Module~\ref{alg:rat-pdes} to obtain a polynomial system. Let $p\coloneq u_{xxx}$ and $d\coloneq u$. First, note that the right-hand side of \eqref{eq:ex-rat-1} is already expressed in the factorized form $p/d$, where the only factorization term of $d$ is $f_1:=u$. Then, we immediately follow step 7 of Module~\ref{alg:rat-pdes} and introduce the auxiliary variable $q\coloneq 1/f_1 = 1/u$ to define the ideal set $\cI := \langle qu - 1\rangle$, from which we construct the Gr{\"o}bner basis $\{qu - 1\}$. Following step 12, we rewrite the original PDE as $u_t = u_{xxx}q$, add the corresponding differential equation of $q$ to the polynomial PDE, and reduce it with respect to the Gr{\"o}bner basis. Then, we obtain the following PDE:
        \begin{align}
        \label{eq:ex-rat21}
            u_t &= u_{xxx}q, \\
        \label{eq:ex-rat22}
            q_t &= -q^3u_{xxx}.
        \end{align}
    
    Now, to find a quadratization for the polynomial PDE described by \eqref{eq:ex-rat21} and \eqref{eq:ex-rat22}, we execute Module~\ref{alg:b&b}, reducing each polynomial multiplication with respect to the defined Gr{\"o}bner basis. With this modification, Module~\ref{alg:b&b} finds the quadratization $\{q^3\}$. Defining $w:=q^3$, we obtain the quadratic PDE system:
    \begin{align*}
        w_t &= w_{xxx}q + 2q_{x}w_{xx} - 14q_{xx}w_{x}, \\
        q_t &= -u_{xxx}w, \\
        u_t &= u_{xxx}q.
    \end{align*}
    Note that if we apply Module~\ref{alg:b&b} to the polynomialized system~\eqref{eq:ex-rat21} and~\eqref{eq:ex-rat22} without executing Module~\ref{alg:rat-pdes}, the best quadratization we would obtain would be $\{q^3, q^5\}$, as the algorithm would not register the definition of $q$. 
    \end{example}

\subsubsection{\texttt{QuPDE} overview}
\label{sec:QuPDEover}

To summarize the construction of \texttt{QuPDE}, Figure~\ref{fig:qupde-diag} presents a flowchart depicting \texttt{QuPDE}'s routine when a nonquadratic PDE is given as input, showing how Module~\ref{alg:mon-decomp}, Module~\ref{alg:verif-quad}, Module~\ref{alg:b&b}, and Module~\ref{alg:rat-pdes} interact to find a quadratization. First, if the PDE system has rational expressions, \texttt{QuPDE} executes Module~\ref{alg:rat-pdes} to polynomialize it. Once there is a polynomial PDE, \texttt{QuPDE} identifies its quadratization candidates by running Module~\ref{alg:mon-decomp}. Then, Module~\ref{alg:b&b} constructs the search tree that determines when Module~\ref{alg:verif-quad} verifies if a set is a quadratization, and runs Module~\ref{alg:mon-decomp} to generate the subproblems. Finally, it produces the best quadratization found for the given PDE. 

    \begin{figure}[H]
        \centering
        \includegraphics[width=\textwidth]{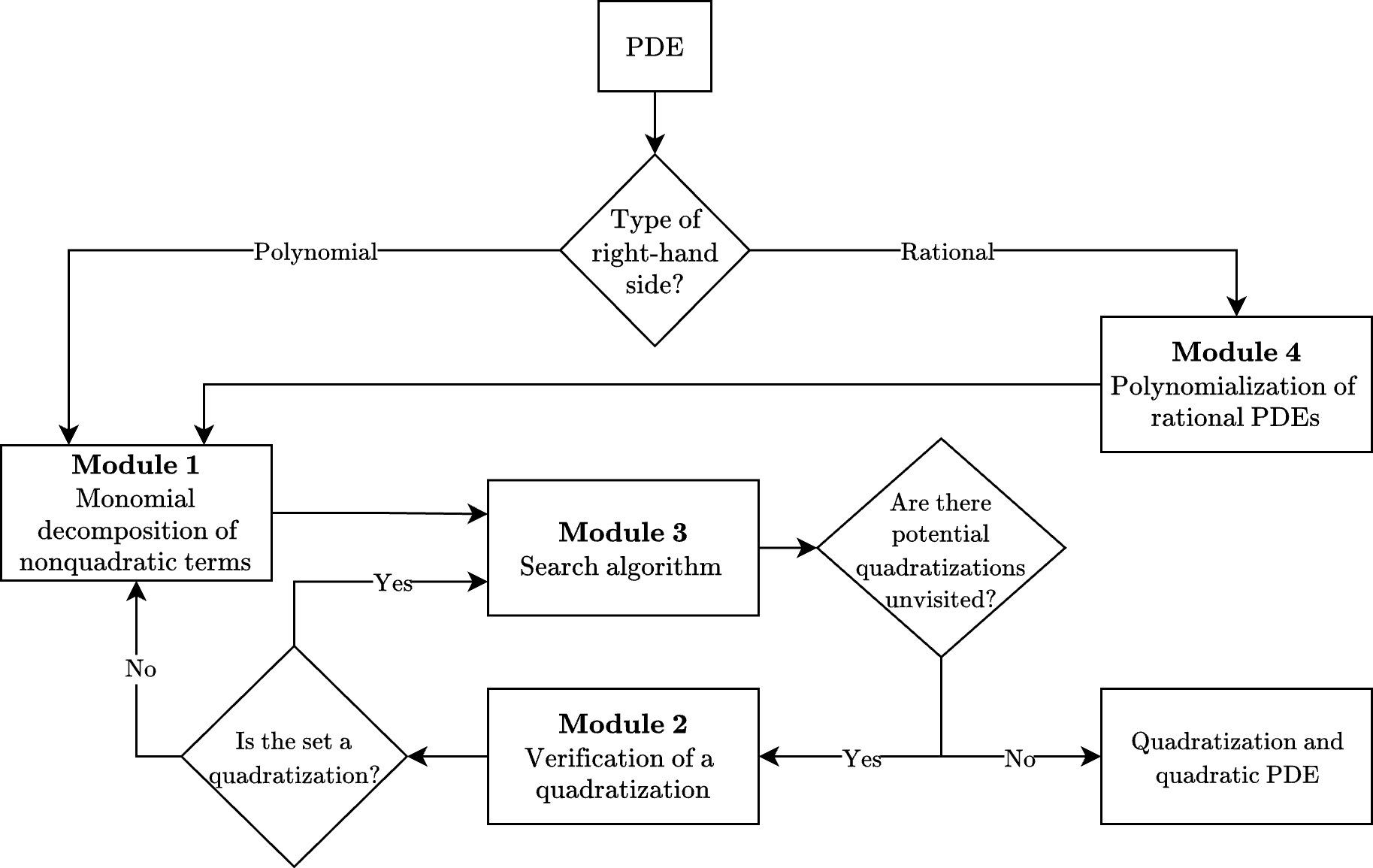}
        \caption{Flowchart for \texttt{QuPDE} to illustrate the interaction between the different modules.}
        \label{fig:qupde-diag}
    \end{figure}

\subsection{Practical aspects of \texttt{QuPDE}}
\label{sec:correct}

Theorem~\ref{thm:existence} shows that it is always possible to find a quadratization with \textnormal{\difforder} $3h$ for a polynomial PDE of the form~\eqref{eq:def-pde}. When this condition is met theoretically but \texttt{QuPDE} does not find a quadratization, it may be because (i) a user-specified value for the \textnormal{\difforder} of quadratization $k$ is not high enough, or (ii) the bound defined for the initial value of PR1 in Module~\ref{alg:b&b}---the maximum depth of the search tree allowed---did not suffice. In either case, following the steps in Proposition~\ref{prop:succ-search} ensures that a quadratization will be found.

    \begin{proposition}
        \label{prop:succ-search}
        For any polynomial and rational PDE of the form~\eqref{eq:def-pde}, the combined application of Module~\ref{alg:mon-decomp}, Module~\ref{alg:verif-quad}, Module~\ref{alg:b&b}, and (if necessary) Module~\ref{alg:rat-pdes} is guaranteed to find a quadratization by following the steps:
        \begin{enumerate}[noitemsep]
        \item Fix initial values for the bound on the auxiliary variables, $N$, and the \textit{\difforder} of quadratization, $k$.
        \item Run \texttt{QuPDE} with the chosen $N$ and $k$.
        \item If the search is unsuccessful, set $k=k+1$ and $N=2\cdot N$, and go back to step 2.
    \end{enumerate}
    \end{proposition}

    \begin{proof}

    Theorem~\ref{thm:existence} above guarantees that spatially one-dimensional evolutionary PDEs of the form~\eqref{eq:def-pde} admit a quadratization of \textnormal{\difforder} $3h$, that is, with solutions in $C^{3h}(\Omega)$. Moreover, \texttt{QuPDE}'s search algorithm is based on the branch-and-bound framework, which implicitly enumerates in a tree all candidate quadratizations within the defined search space, guaranteeing that the optimal solution in the tree is found while avoiding an exhaustive search \cite{B&B}. Therefore, \texttt{QuPDE} can always find a quadratization by increasing both the number of auxiliary functions and the \textnormal{\difforder} $k$ up to $3h$.
    \end{proof}

\texttt{QuPDE} is a symbolic search algorithm that only knows the smoothness of a PDE solution through the user-specified \textnormal{\difforder} of quadratization $k$. Consequently, (i) \texttt{QuPDE} can find quadratizations for PDEs with $u\in C^{\ell}(\Omega)$, where $\ell<3h$;
and (ii) \texttt{QuPDE} could output quadratizations for a PDE that assume a solution is $\ell+j$ times differentiable, even if the original solution of the PDE is only in $C^\ell(\Omega)$. An example of case (ii) is shown below.
\begin{example}[\texttt{QuPDE}'s smoothness assumption]
    Consider the PDE $u_t = u_{x}^2u$ with $\ell=1$, i.e., $u\in C^{1}(\Omega)$. We define the auxiliary variable $w:= u_xu$ and set $k=2$ to calculate $w_x=u_{xx}u+u_{x}^2$. The set $\{w=u_xu\}$ is a quadratization for this PDE and yields the quadratic system given by $u_t = u_xw$ and $w_t = 2w_xw$.
    Since $k$ was set to a value larger than $\ell$, this variable transformation assumes that $u$ is a solution $\ell+1$ times differentiable, where the expression for $w_x$ has a spatial derivative of $u$ of order two.  
\end{example}

\section{Benchmarks and Results}
\label{sec:bench-res}

This section presents the results of using \texttt{QuPDE} to find quadratizations for fourteen nonquadratic spatially one-dimensional PDEs from the literature of diverse areas such as fluid mechanics, space physics, chemical engineering, and biological processes. Section~\ref{sec:benchmark} lists the nonquadratic PDEs, Section~\ref{sec:res} shows the results of searching for a quadratization for each of them, and Section~\ref{sec:res-qbee} presents how \texttt{QuPDE} compares with \texttt{QBee} \cite{quadODE}, the only alternative software that offers an algorithmic solution for the PDE quadratization problem via the detour of semi-discretization.

\subsection{PDE models}
\label{sec:benchmark}

The governing equations for each PDE benchmark are shown below. We omit the initial and boundary conditions since the quadratizations are found based solely on the symbolic form of the PDEs, c.f. Remark~\ref{rem:scope}.
    
    \paragraph{Dym equation} 
    The Harry Dym equation \cite{KUMAR2013111} 
    arises in several areas, such as the analysis of the Saffman-Taylor problem with surface tension \cite{Kadanoff1990} and acoustic theory \cite{Zakharov2007}. This equation is completely integrable and represents a system in which dispersion and nonlinearity are coupled:
    \begin{equation}
    \label{eq:dym}
        u_{t}=u^{3}u_{xxx}.
    \end{equation}

    \paragraph{Arrhenius-type reaction model}
    
    The non-adiabatic tubular reactor model in \cite{Heinemann1981} describes species concentration and temperature evolution in a single reaction through
    \begin{align}
    \label{eq:tubular-exp-1}
        u_t &= \epsilon u_{xx} - u_x - \mathcal{D}uf(v), \\
    \label{eq:tubular-exp-2}
        v_t &= \epsilon v_{xx} - v_x - \beta(v - \theta_\text{ref}) + \mathcal{B}\mathcal{D}uf(v),
    \end{align} 
    where $f(v)=e^{\gamma-\frac{\gamma}{v}}$ is an Arrhenius nonlinear reaction term describing the reaction rate. Here, $\mathcal{D}$ is the Damk{\"o}hler number, $\epsilon = \frac{1}{Pe}$, with $Pe$ the P\`eclet number and $\mathcal{B}, \beta, \gamma, \theta_\text{ref}$ are known constants of the system. 

    To use \texttt{QuPDE} on this model, we apply a variable transformation to~\eqref{eq:tubular-exp-1} and \eqref{eq:tubular-exp-2} by introducing the auxiliary variable $y = e^{\gamma-\frac{\gamma}{v}}$. Then, we write the rational PDE system
    \begin{align}
    \label{eq:tubular-exp-pol-1}
        u_t & = \epsilon u_{xx} - u_x - \mathcal{D}uy, \\
        \label{eq:tubular-exp-pol-2}
        v_t & = \epsilon v_{xx} - v_x - \beta(v - \theta_\text{ref}) + \mathcal{B}\mathcal{D}uy,\\
        \label{eq:tubular-exp-pol-3}
        y_t & = \dfrac{\gamma}{v^2}y\left(\epsilon v_{xx} - v_x - \beta(v - \theta_\text{ref}) + \mathcal{B}\mathcal{D}uy\right).
    \end{align} 
    
    \paragraph{Polynomial reaction model}
    
    An approximation via Taylor expansion of the non-adiabatic tubular reactor model \cite{KW2019_balanced_truncation_lifted_QB} described in \eqref{eq:tubular-exp-1} and \eqref{eq:tubular-exp-2} is
    \begin{align}
    \label{eq:tubular-pol-1}
        u_t &= \epsilon u_{xx} - u_x - \mathcal{D}f_d(v), \\
        \label{eq:tubular-pol-2}
        v_t & = \epsilon v_{xx} - v_x - \beta(v - \theta_\text{ref}) + \mathcal{B}\mathcal{D}uf_d(v),
    \end{align} 
    where $f_d(v)$ is a Taylor series approximation of the exponential Arrhenius term, i.e., $f_d(v) = c_0 + c_1v + c_2v^2 + \dots + c_{d}v^d$. Section~\ref{sec:res} shows results with approximations of polynomial degrees three, four, and five.
    
    \paragraph{Modified KdV equation}  
    The modified Korteweg-De Vries equation \cite{kdv} is a model for the study of weakly nonlinear long waves, incorporating leading nonlinearity and dispersion. It also describes surface waves of long wavelength and small amplitude in shallow water, and has the form
    \begin{equation}
    \label{eq:kdv}
        u_t = au^2 u_x - u_{xxx}, 
    \end{equation}
    where $a>0$.
    
    \paragraph{Solar wind model} 
    The HUX (Heliospheric Upwinding eXtrapolation) model \cite{solarwind_2023} is a two-dimensional time-stationary model that predicts the heliospheric solar wind speed via the PDE
    \begin{equation}
    \label{eq:solar-wind}
        u_r = \frac{\Omega_{\rm{rot}}u_\phi}{u}, 
    \end{equation}
    where $\Omega_{\rm{rot}}$ is the angular frequency on the Sun's rotation evaluated at a constant Carrington latitude. In our notation, $\phi$ is the spatial variable and $r$ plays the role of time.
    
    \paragraph{Schl{\"o}gl model}
    Schl{\"o}gl's model \cite{schlogl_model} is a simple example of a chemical reaction system that exhibits bistability. We consider a version of this model without input functions:
    \begin{equation}
        \label{eq:schlogl}
        u_t = u_{xx} - k(u - u_1)(u - u_2)(u - u_3),
    \end{equation} 
    with $k\geq 0$ and real numbers $u_1 < u_2 < u_3$.
    
    \paragraph{Euler equations}
    The Euler equations \cite{EULEREQ} are derived from the physical principles of conservation of mass, momentum, and energy and describe an inviscid flow through the equations
    \begin{align}
        \label{eq:euler}
        \rho_t & = -u\rho_x - \rho u_x, \\
        u_t & = -u_xu - \frac{p_x}{\rho}, \\
        p_t & = -\gamma u_xp - up_x,
    \end{align} 
    where $\gamma$ is the ratio of specific heats. 
    
    \paragraph{FitzHugh-Nagamo system} 
    The FitzHugh-Nagamo system \cite{FHN} is a simplified neuron model of the Hodgkin-Huxley model, which describes activation and deactivation dynamics of a spiking neuron via the PDE system
    \begin{align}
    \label{eq:fhn}
        v_t & = \epsilon v_{xx} + \dfrac{1}{\epsilon}v(v - 0.1)(1 - v) - \dfrac{1}{\epsilon}u + \dfrac{1}{\epsilon}c, \\
        u_t & = bv - \gamma u + c,
    \end{align} 
    where $\epsilon, b, \gamma,c$ are real and positive constants. 
    
    \paragraph{Schnakenberg equations}
    The Schnakenberg equations \cite{schnakenberg} are evolution equations for reaction-diffusion systems with cross-diffusion, given by
    \begin{align}
    \label{eq:schnakenberg}
        u_t & = D_u u_{xx} + D_{uv}v_{xx} + k_1a_1 - k_2u + k_3u^2v, \\
        v_t & = D_v v_{xx} + D_{vu}u_{xx} + k_4b_1 - k_3u^2v,
    \end{align} 
    where $D_u > 0$, $D_v > 0$, $D_{uv}$ and $D_{vu}$ are diffusion and cross-diffusion coefficients, and $a_1$, $b_1$, $k_1$, $k_2$, $k_3$, $k_4$ are all positive constants.
    
    \paragraph{Brusselator system}
    The Brusselator \cite{brusselator} models morphogenesis and pattern formation in chemical reactions through the equations
    \begin{align}
        \label{eq:brusselator}
        u_t&  = d_1 u_{x} + \lambda (1 - (b + 1)u + bu^2v), \\
        v_t&  = d_2 v_{x} + \lambda a^2(u - u^2v),
    \end{align} 
    where $d_1$, $d_2$, $\lambda$, $a$, $b$ are positive constants.
    
    \paragraph{Allen-Cahn equation}  
    The Allen-Cahn equation \cite{Allen1979} describes the motion of antiphase boundaries in crystalline solids. It was proposed as a simple model for the process of phase separation of a binary alloy at a fixed temperature:
    \begin{equation}
        \label{eq:cahn}
        u_t = u_{xx} + u - u^3.
    \end{equation}
    
    \paragraph{One-dimensional nonlinear heat equation}  
    The nonlinear heat equation \cite{Bandle1998} presents a rich mathematical structure for studying blow-up phenomena, where solutions cease to exist after a finite time. The equation takes the form
    \begin{equation}
        \label{eq:heat}
        u_t = u_{xx} + u^p,
    \end{equation}
    with $p>1$. For our benchmark, we performed tests with $p=6$.

\subsection{Results}
\label{sec:res}
The results presented in this section were obtained using our software implementation of \texttt{QuPDE} \cite{qupde} version 0.1.0\footnote{this version is available at \url{https://doi.org/10.5281/zenodo.18750665}.
 Current version and support available at \url{https://github.com/albaniolivieri/pde-quad.git}}, which is developed as a Python package
built using the \texttt{Sympy} library \cite{sympy}. All experiments were run on an Apple M4 Pro personal laptop with clock speed of 4.5 GHz, 48GB RAM, and macOS Tahoe 26.2 with Python version 3.14.2. In Table~\ref{tab:gen-b&b}, we show the results of running \texttt{QuPDE} on each example described in Section~\ref{sec:benchmark}. For each PDE, we follow the procedure in Proposition~\ref{prop:succ-search} to find the lowest \textnormal{\difforder} at which \texttt{QuPDE} finds a quadratization, fixing as the initial value the order of the original PDE. Moreover, we use H3 as the sorting heuristic since it performed best overall. 
    \begin{table}[H]
    \centering
    \caption{Results of applying \texttt{QuPDE} to the fourteen PDE problems from above. Listed are the auxiliary variables the algorithm introduced, the execution time, and the number of nodes in the search tree that the algorithm traversed. Examples are ordered by CPU time.}
    \label{tab:gen-b&b}
    \begin{tabular}{l|c|c|c}
        PDE & Quadratization variables & CPU time & Nodes traversed \\ \hline
        Solar wind model & $1/u$ & 8.8 $\pm$ 0.1 [ms] & 1 \\
        Allen-Cahn equation & $u^2$ & 22.4 $\pm$ 0.1 [ms] & 3\\ 
        Schl{\"o}gl model & $u^2$ & 43.8 $\pm$ 0.2 [ms] & 3 \\ 
        Modified KdV & $u^2$ & 54.2 $\pm$ 0.9 [ms] & 4 \\
        Euler equations & $1/\rho$ & 54.6 $\pm$ 1.1 [ms] & 1 \\ 
        FHN system & $v^2$ & 115.1 $\pm$ 1.1 [ms] & 3 \\ 
        Brusselator system & $u^2$, $uv$ & 129.6 $\pm$ 1.2 [ms] & 8\\ 
        Nonlinear heat equation & $u^2$, $u^4$, $u^5$ & 205.1 $\pm$ 2.5 [ms] & 27  \\ 
        Schnakenberg equations & $uv$, $u^2$ & 445.0 $\pm$ 4.1 [ms] & 8 \\ 
        Dym equation & $u^3$, $u_{x}^2u$ & 628.9 $\pm$ 5.5 [ms] & 21 \\
        Polynomial reaction ($d = 3$) & $uv$, $v^2$, $v^2u$, $v^3$ & 9.9 $\pm$ 0.1 [s] & 69 \\ 
        Polynomial reaction ($d = 4$) & $uv$, $v^2$, $v^2u$, $v^4$, $v^3u$ & 62.5 $\pm$ 0.4 [s] & 305 \\ 
        Arrhenius-type reaction & \makecell{$1/v$, $1/v^2$, $uv$, $uy/v$,\\ $uy/v^2$, $y/v$, $y/v^2$} & 179.1 $\pm$ 0.4 [s] & 491 \\
        Polynomial reaction ($d = 5$)& $uv$, $v^2$, $v^3$, $v^3u$, $v^4u$, $v^5$ & 636.1 $\pm$ 4.6 [s] & 2107 \\
         
    \end{tabular}
\end{table}

Table~\ref{tab:gen-b&b} shows that, for most of the examples, \texttt{QuPDE} produces the best quadratization found within the order of milliseconds, even for PDEs with rational nonlinearities. Moreover, in nine out of fourteen cases, it finds quadratizations of only one or two auxiliary variables. The last four examples, in which \texttt{QuPDE} produces the largest quadratizations and requires the longest computation time, correspond to the highly nonlinear models for the non-adiabatic tubular reactor. Two of these systems are polynomials of degree five and six, since the degree of $f_d$ in~\eqref{eq:tubular-pol-1} and~\eqref{eq:tubular-pol-2} is four and five, respectively; and the Arrhenius-type reaction model corresponds to the rational PDE system in~\eqref{eq:tubular-exp-pol-1}, \eqref{eq:tubular-exp-pol-2}, and~\eqref{eq:tubular-exp-pol-3}. It is worth noting that for the last three examples, the time it takes for \texttt{QuPDE} to find a quadratization, compared to the time for verification, varies considerably. For instance, in the Arrhenius-type reaction model, \texttt{QuPDE} finds the first quadratization within 17 seconds, while the remaining 166 seconds are spent searching for a lower-order quadratization alternative through Module~\ref{alg:b&b}. Moreover, to the best of our knowledge, the quadratizations presented here are the first for the nonlinear heat equation and the polynomial reaction models with $d=4$ and $d=5$.

\subsection{Comparison to the state-of-the-art algorithm \texttt{QBee}}
\label{sec:res-qbee}

The \texttt{QBee} algorithm for ODEs \cite{bychkov2023exact} offers the only available alternative for finding quadratizations of PDEs via the detour of semi-discretization. Therefore, we compare \texttt{QuPDE} with results obtained from \texttt{QBee} \cite[Algorithm~5.3]{bychkov2023exact} on published semi-discretized forms of four PDE systems: the solar wind model \cite{solarwind_2023}, the Allen-Cahn equation \cite{ALLENCAHNDISC}, and two of the non-adiabatic tubular reactor models: the polynomial approximation with $d=3$ \cite{KW2019_balanced_truncation_lifted_QB} and the Arrhenius-type reaction \cite{zhou2012model}. This comparison is shown in Table~\ref{tab:qbee}. 
\begin{table}[H]
        \centering
        \caption{Comparison between \texttt{QBee} dimension-agnostic quadratization and the proposed \texttt{QuPDE}. Listed are the execution times and the order of the best quadratization found for each PDE.}
        \label{tab:qbee}
        \begin{tabular}{l|c | c | c | c}
            \multirow{2}{7em}{PDE} & \multicolumn{2}{c|}{\texttt{QBee}} & \multicolumn{2}{c}{\texttt{QuPDE}} \\ 
            \cline {2-5}&  CPU time & Quad. order & CPU time & Quad. order \\
            \hline 
             Solar wind model & 0.47 [s] & 4 & 0.01 [s] & 1 \\ 
             Allen-Cahn equation & 0.06 [s] & 1 & 0.02 [s] & 1 \\ 
             Polynomial reaction $(d=3)$ & 0.54 [s] & 4 & 9.94 [s] & 4 \\ 
             Arrhenius-type reaction & 76.77 [s] & 8 & 179.13 [s] & 6
        \end{tabular}
\end{table}

In Table~\ref{tab:qbee}, we see significant differences in the results of the two algorithms. In the solar wind example and the Arrhenius-type reaction model, our algorithm outperforms \texttt{QBee} in the number of variables introduced. For the solar wind model, \texttt{QBee} algorithmically produces a quadratization with four auxiliary variables. However, the semi-discretization of~\eqref{eq:solar-wind} introduces a logarithmic expression, requiring an additional auxiliary variable $w:=\ln(u)$ from the polynomialization procedure. Consequently, using the \texttt{QBee} approach yields to a quadratization of five auxiliary variables: $\{w, 1/u, w/u, \tilde{u}/u,\tilde{w}/u\}$, where $\tilde{u}$ and $\tilde{w}$ represent the nonzero off-diagonal elements of the discretization system matrix, see \cite[Section~7.2]{bychkov2023exact} for more details. This adds five differential equations to~\eqref{eq:solar-wind}, 
whereas \texttt{QuPDE} requires a single auxiliary variable $1/u$ through the application of Module~\ref{alg:rat-pdes}. In the Arrhenius-type reaction model, \texttt{QBee} yields a quadratization with two more auxiliary variables than \texttt{QuPDE}: $\{1/v, 1/v^2, uy/v^2, uy/v, \tilde{v}/v, \tilde{u}/v, \tilde{v}/v^2, \tilde{u}/v^2\}$, where the variables $\tilde{u}, \tilde{v}$ correspond to the discrete representation of the PDE system in~\eqref{eq:tubular-exp-1} and~\eqref{eq:tubular-exp-2}. For the polynomial reaction model and the Allen-Cahn equation, \texttt{QuPDE} and \texttt{QBee} produce the same quadratization. On the other hand, while our algorithm takes longer to find the same quadratization for the polynomial reaction model, its time performance is similar to that of \texttt{QBee} when finding quadratizations for the Allen-Cahn equation and the solar wind model, finding a lower-order quadratization for the latter.

We emphasize once again that the state-of-the-art method \texttt{QBee} is limited to finding quadratizations for semi-discretized PDEs that are affine in their spatial derivatives. Discretizing nonlinear PDEs is a complex and often cumbersome task, involving many choices in the process that can result in different semi-discretized forms. It also demands a solid understanding of the PDE's structural properties and its domain. Moreover, requiring users to provide a pre-discretized PDE makes the process of finding quadratizations significantly less user-friendly. As a result, only a limited number of examples are suitable for \texttt{QBee}. On the other hand, the proposed \texttt{QuPDE} offers a user-oriented design that works directly on the symbolic form of any polynomial and rational PDE.

\section{Conclusions}
\label{sec:conclusions}

The concept of variable transformation has been widely explored to bring nonquadratic PDEs to a simpler---here quadratic---form, where control, analysis, and model reduction become easier. While one can find quadratizations for ODE systems using two freely available computational tools \cite{quadODE, hemery2020complexity}; for PDEs, the two options available are either to derive a quadratization by hand or to take the detour of semi-discretization of the PDE and use the ODE quadratization algorithm in \cite{bychkov2023exact}. Both approaches are suboptimal and, in the latter case, require expert knowledge in PDE discretization and only applies to a subclass of PDEs.

The main contribution of the present work is the algorithm \texttt{QuPDE} for finding quadratizations of spatially one-dimensional evolutionary PDEs with nonlinearities in their right-hand side, along with a freely available software implementation \cite{qupde}. This algorithm rests on an initial theoretical groundwork that we establish in this paper: a rigorous definition of PDE quadratization, a theoretical result relating the solutions of the lifted system to those of the original PDE for the modified KdV equation, a proof that every polynomial PDE admits a quadratization under certain regularity conditions, and a demonstration that finding quadratic transformations for PDEs via auxiliary variables is an NP-hard problem.
The \texttt{QuPDE} algorithm works directly with the PDE formulation of the quadratization problem for any polynomial or rational spatially one-dimensional PDE. To accomplish this, \texttt{QuPDE} employs a symbolic computation-based approach that relies on an optimization algorithm using the Branch-and-Bound framework to make the search for a low-order quadratization feasible. We speed up the search by implementing pruning rules, sorting heuristics, and computational strategies to reduce the number of traversed nodes. \texttt{QuPDE} also includes a polynomialization step for rational functions, a tool to handle a wider range of nonlinear models.

We tested \texttt{QuPDE} on fourteen PDE models from diverse areas such as fluid mechanics, space physics, chemical engineering, and biological processes. In each case, \texttt{QuPDE} delivered a low-order quadratization for the benchmarks presented. As demonstrated in our results, \texttt{QuPDE} has (i) discovered new quadratizations that needed fewer auxiliary variables than previously reported in the literature, and, most importantly, does not impose additional requirements such as semi-discretization; and (ii) found entirely new quadratizations for systems that had not been brought to quadratic form before, such as the nonlinear heat equation and the high-degree polynomial models of the tubular reactor model.

This work motivates multiple future research directions. First, there is a need to develop theory related to the solution, properties, and the initial boundary value problem of the quadratic PDE representations. Second, designing and implementing a polynomialization algorithm that extends PDE quadratization to nonpolynomial systems beyond rational functions would be an extension worth exploring; the nonlinear sine-Gordon equation \cite{Barone1971} illustrates the advantages of this technique, which becomes quadratic after polynomialization. Third, another promising direction is the development of methods to further reduce the number of equations of the resulting quadratic PDE systems. Fourth, the algorithm could be extended to guarantee globally optimal quadratizations and to improve efficiency through enhanced search strategies, e.g., random tree search, and additional physics-informed pruning rules. Fifth, since many applications require high-dimensional PDE models, it would be worthwhile to generalize the theory and algorithmic tools to address PDEs in a higher spatial dimension. Lastly, an interesting extension of \texttt{QuPDE} would be to automate the search for structure-preserving quadratizations, building on the ideas proposed in \cite{Sharma2025energy, Sharma2026} for energy-conserving systems.

\subsection*{Acknowledgments}
The authors would like to thank Gonzalo Navarro for his valuable input on the practical aspects of \texttt{QuPDE}. B.K. and A.O. are supported in part through NSF-CMMMI award 2144023. G.P. was partially supported by the INS2I PANTOMIME project and the French ANR-22-CE48-0016 NODE project.
\bibliography{references}

\bibliographystyle{unsrt}

\end{document}